\documentclass[11pt]{article}

\usepackage[english]{babel}
\usepackage[latin1]{inputenc}
\usepackage{lmodern}
\usepackage[T1]{fontenc}

\usepackage[affil-sl]{authblk}
\usepackage{amssymb, amsmath, amsthm, blkarray}
\usepackage{graphicx, subfig, tikz}
\usepackage{dsfont}
\usepackage{natbib}

\usepackage[a4paper,margin=30mm]{geometry}

\usepackage{marginnote}
\usepackage{soul}

\usepackage{pgfplots,pgfplotstable,booktabs,colortbl}
\usetikzlibrary{patterns}
\usepgfplotslibrary{fillbetween}
\pgfplotsset{compat=newest}
\pgfplotstableread{data/vxx_vix_implied.dat}\vixvxx
\pgfplotstableread{data/vxx_mean.dat}\vxxmean
\pgfplotstableread{data/vxx_volvol.dat}\vxxvolvol
\pgfplotstableread{data/histogram.dat}\rhohist

\definecolor{mylightyellow}{rgb}{1,1,.8}
\definecolor{mylightgreen}{rgb}{.8,1,.8}
\definecolor{mydarkred}{RGB}{178,34,34}
\definecolor{mydarkgreen}{RGB}{34,139,34}
\definecolor{mydarkblue}{RGB}{72,61,139}
\definecolor{mydarkyellow}{RGB}{218,165,32}

\definecolor{Bcolor}{rgb}{0,0.7,0}
\definecolor{Gcolor}{rgb}{1,0,0.5}

\usepackage[hyphens,spaces,obeyspaces]{url}
\usepackage[]{hyperref}
\hypersetup{
 colorlinks=true,breaklinks=true,
 urlcolor= mydarkblue,linkcolor= mydarkblue,citecolor= mydarkgreen,
 pdfauthor={Grasselli, M., Mazzoran, A., Pallavicini, A.}
}

\newtheorem{Proposition}{Proposition}[part]

\newtheorem{Lemma}{Lemma}[part]

\newtheorem{Remark}{Remark}[part]

\makeatletter
\newcommand\footnoteref[1]{\protected@xdef\@thefnmark{\ref{#1}}\@footnotemark}
\makeatother

\newcommand{\Ex}[2]{\mathbb{E}_{#1}\!\left[\,#2\,\right]}

\newcommand{\ExC}[3]{\mathbb{E}_{#1}\!\left[\left.\,#2\,\right|\,#3\,\right]}

\newcommand{\ind}[1]{1_{\{#1\}}}

\DeclareMathOperator{\Tr}{Tr}
\newcommand{\trace}[1]{{\Tr\left[{#1}\right]}}

\begin{document}

\title{A general framework\\for a joint calibration of VIX and VXX options}

\author[a,b]{Martino Grasselli\thanks{grassell@math.unipd.it}}
\author[a]{Andrea Mazzoran\thanks{andreamazzoran@hotmail.it}}
\author[c,d]{Andrea Pallavicini\thanks{andrea.pallavicini@intesasanpaolo.com}}

\affil[a]{\small Department of Mathematics, University of Padova, Via Trieste 63 Padova 35121, Italy.}
\affil[b]{\small Devinci Research Center, L\'{e}onard de Vinci P{\^ o}le Universitaire, 92 916 Paris La D\'{e}fense, France.}
\affil[c]{\small Department of Mathematics, Imperial College, London SW7 2AZ, United Kingdom.}
\affil[d]{\small Financial Engineering, Intesa Sanpaolo, largo Mattioli 3 Milano 20121, Italy.}

\date{\today}

\maketitle

\begin{abstract}
We analyze the VIX futures market with a focus on the exchange-traded notes written on such contracts, in particular we investigate the VXX notes tracking the short-end part of the futures term structure. Inspired by recent developments in commodity smile modelling, we present a multi-factor stochastic-local volatility model that is able to jointly calibrate plain vanilla options both on VIX futures and VXX notes, thus going beyond the failure of purely stochastic or simply  local volatility models. We discuss  numerical  results on real market data by highlighting the impact of model parameters on implied volatilities.
\end{abstract}

\bigskip

 \medskip\\
\noindent \textbf{AMS classification codes:} 65C05, 91G20, 91G60. \medskip


\section{Introduction}

In recent years, it has become increasingly common to consider volatility as its own asset class. As derivatives on the VIX index, created by the CBOE\footnote{See the white paper of CBOE in 2014 titled \textit{CBOE Volatility Index}, available at \url{https://www.cboe.com/micro/vix/vixwhite.pdf.}} in 1993, may be inaccessible to non-institutional players, mainly due to the large notional sizes of the contracts as currently designed, some exchange-traded notes (ETN) on the VIX were introduced. In $2009$, Barclays launched the first two ETN's on the VIX: VXX\footnote{\label{footnote1}See the VXX Prospectus of Barclays available at \url{https://www.ipathetn.com/US/16/en/instruments.app?searchType=text&dataType=html&searchTerm=VXX}} and VXZ. The VXX ETN was the first exchange-traded product (ETP) on VIX futures, issued shortly after the inception of the VIX futures indices. The VXX is a non-securitized debt obligation, similar to a zero-coupon bond, but with a redemption value that depends on the level of the S\&P 500 (SPX) VIX Short-Term Futures Total Return index (SPVXSTR). The SPVXSTR tracks the performance of a position in the nearest and second-nearest maturing VIX futures contracts, which is rebalanced daily to create a nearly constant $1$-month maturity. In other words, VXX mimics  the behaviour of the VIX, namely of the 30-day forward-looking volatility, using a replicable strategy based on futures on VIX.  Since $2009$ then, ETN's on the VIX flourished: there currently exist more than thirty of them, with several billion dollars in market caps and daily volumes (see~\cite{alexander2013volatility} for a comprehensive empirical study on VIX ETN's).

There is a considerable literature for the VIX, and attempts to model VIX can be divided in two strands: on the one hand, the consistent-pricing approach models the joint (risk-neutral) dynamics of the S\&P 500 and the VIX with, most often, the aim of pricing derivatives on the two indices in a consistent manner, on the other hand, the stand-alone approach directly models the dynamics of the VIX.

In the former case, we can include the work of~\cite{cont2013consistent} and references therein, while more applications can be found in the work of~\cite{gehricke2018modeling}, where they propose a consistent framework for S\&P 500 (modelled with a Heston-like specification) and the VIX derived from the square root of the variance swap contract. On top of that, they construct the VXX contract from the Prospectus and show that the roll yield of VIX futures drives the difference between the VXX and VIX returns on time series. More recently, in the paper of~\cite{gatheral2020quadratic}, a joint calibration of SPX and VIX smile has been  successfully obtained using  the quadratic rough Heston model, while~\cite{guyon2020joint}  builds up a non-parametric discrete-time model that jointly and exactly calibrates to the prices of SPX options, VIX futures and VIX options. We also mention the work of~\cite{guyon2020inversion}, where the author investigates conditions for the existence of a continuous model on the S\&P 500 index (SPX) that calibrates the full surface of SPX and VIX implied volatilities, and the paper of~\cite{de2015linking}, where the authors bound VIX options from vanilla options and VIX futures, which leads to a new martingale optimal transport problem that they solve numerically.

In the latter case, the first stand-alone model for the VIX was proposed by Whaley, see~\cite{whaley1993derivatives}, where he models the spot VIX as a geometric Brownian motion, hence enabling to price options on the VIX in the Black and Scholes framework, but failing to capture the smile effect and the mean-reversion of volatility processes, which is necessary in order to fit the term structure of futures prices (\cite{schwert1990stock}, \cite{pagan1990alternative} and~\cite{schwert2011stock}). We cite also the attempt of~\cite{Drimus2013} based on local volatility. Since then, numerous stand-alone models have been proposed for the VIX. We refer to~\cite{bao2012pricing} and~\cite{lin2013vix} for a detailed review on that strand of literature.

The literature on VXX is limited. To our knowledge, there are only a few theoretical studies which propose a unified framework for VIX and VXX. Existing studies on the VXX (\cite{bao2012pricing}) typically model the VXX in a stand-alone manner, which does not take into account the link between VIX and VXX. As a consequence, the loss for VXX when the VIX is in contango cannot be anticipated in such models. In the recent paper of~\cite{gehricke2020implied}, they analyze the daily implied volatility curves of the VXX options market alone (where they consider the implied volatility as a function of the moneyness through quadratic polynomial regressions), providing a necessary benchmark for developing a VXX option pricing model. We also mention the work of~\cite{grasselli2018vix}, where the authors link the properties of VXX to those of the VIX in a tractable way. In particular, they quantify the systematic loss observed empirically for VXX when the VIX futures term-structure is in contango and they derive option prices, implied volatilities and skews of VXX from those of VIX in infinitesimal developments. The affine stochastic volatility  model of~\cite{grasselli2018vix} exploits the FFT methodology in order to price futures and options on VIX, while a Monte Carlo simulation is performed in order to compare the resulting smile on the VXX implied by the calibrated model with the one provided by real market data. The results are not completely satisfactory, which is not surprising, since the VXX can be seen as an exotic path-dependent product on VIX, so that the information regarding the smile on vanillas (on VIX) are well known to be not enough to recover the whole information concerning path dependent products in a purely stochastic volatility model (see for instance~\cite{bergomi2015stochastic}).

In the present paper we extend the pure stochastic volatility model for the VIX and we show that  a stochastic local-volatility model (SLV) is able to jointly calibrate VIX and VXX options in a consistent manner. The same methodology can be easily extended to other pairs of futures contracts and notes. SLV models are the de-facto industry standard in option pricing since they allow to join the advantages of stochastic volatility models with a good fit of plain options. They were first introduced in~\cite{Lipton2002}, while a first efficient PDE-based calibration procedure was presented in~\cite{Ren2007}. In the present paper, since we investigate multi-factor models, we follow the calibration approach of~\cite{guyon2012being}, which is based on a Monte Carlo simulation. The SLV model we introduce for the VIX futures starts from the work of~\cite{nastasi2020smile}, where the problem of modelling futures smiles for commodity assets in a flexible but parsimonious way is discussed. In particular, here we investigate a multi-factor extension of such model which is tailored to describe VIX futures prices, and to allow a joint calibration of VIX and VXX plain vanilla options.

The key feature of our model is the possibility to decouple the problems of calibrating VIX plain vanilla options, which are recovered by a suitable choice of the local-volatility function, from the more complex task of calibrating the plain vanilla options on the VXX notes. We recall that in our approach the VXX notes are exotic path-dependent product written on VIX futures. We solve the second problem by introducing a local-correlation function between the factors driving the futures prices, and by properly selecting the remaining free dynamics parameters. The calibration of the local correlation function is inspired by the work of~\cite{Guyon2017}. Quite remarkably, we find that the stochastic component of the volatility is not redundant, namely we need a truly stochastic-local volatility model in order to be able to fit both VIX and VXX markets using real market quotes.

\medskip

The paper is organized as follows. In Section~\ref{localvol} we give a brief overview of SLV models, highlighting their properties and the main advantages in using these models. In Section~\ref{ETPstrategies}, we present the general setting of our model. We start from the evolution of a general ETP and their underlying futures contracts, then  we describe how to calibrate futures prices and options on futures. In Section~\ref{An Example with VIX and VXX} we adapt the general framework of the previous section to the case of the VXX ETN strategy. In Section~\ref{Numerical Results} a calibration exercise based on real data along with numerical results and comparisons are presented, involving a joint fit and a parameters sensitivity analysis. In Section~\ref{conclusion} we draw some conclusion and remarks.

\medskip

Needless to say, the views expressed in this paper are those of the authors and do not necessarily represent the views of their institutions.

\section{Stochastic local-volatility models}
\label{localvol}
 
One of the main advantages of using local-volatility models (LV) is their natural modelling of plain-vanilla market volatilities. Indeed, a LV model can be calibrated with extreme precision to any given set of arbitrage-free European vanilla option prices. LV model were first introduced by~\cite{dupire1994pricing} and~\cite{derman1994riding}. Although well-accepted, LV models have certain limitations, for example, they generate flattening implied forward volatilities, see~\cite{rebonato1999volatility}, that may lead to a mispricing of financial products like forward-starting options.

On the other hand, stochastic volatility (SV) models, like the well known Heston model, see~\cite{heston1993closed}, are considered to be more accurate choices for pricing forward volatility sensitive derivatives, see~\cite{gatheral2011volatility}. Although the SV models have desired features for pricing, they often cannot be very well calibrated to a given set of arbitrage-free European vanilla option prices. In particular, the accuracy of the Heston model for pricing short-maturity options in the equity market is typically unsatisfactory.

One  possible improvement to the previous issues is considering stochastic local-volatility models (SLV), that take advantages of both LV and SV models properties. SLV models, going back to~\cite{Jex1999} and~\cite{Lipton2002}, are still an intricate task, both from a theoretical as well as a practical point of view. The main advantage of using SLV models is that one can achieve both a good fit to market series data and in principle a perfect calibration to the implied volatility smiles. In such models the discounted price process $(S_t)_{t \geq 0}$ of an asset follows the stochastic differential equation (SDE) 
\begin{equation}
\label{SLV price dyanmics}
    d S_t = S_t \ell(t, S_t) \sigma_t \,dW_t,
\end{equation}%
where $(\sigma_t)_{t \geq 0}$ is a process taking positive values, and $\ell(t, K)$ is a sufficiently regular deterministic function, the so-called leverage function, depending on time and on the current value of the underlying asset. The Brownian motion $W_t$ is possibly correlated with the noise driving the process $\sigma$. Obviously, SV and LV models are recovered by setting $\ell(t,K) \equiv 1$ or $\sigma \equiv 1$, respectively. At this stage, the stochastic volatility process (which in our specification is given by the product $(\ell \sigma_t)_{t \geq 0}$) can be very general.

The leverage function $\ell(t,K)$ is the crucial part in this model. It allows in principle to perfectly calibrate the market implied volatility surface. In order to achieve this goal, $\ell$ needs to satisfy the following relation:
\begin{equation}
\label{leverage computation}
    \ell^2(t,K) = \frac{\eta^2(t,K)}{\ExC{0}{\sigma_t^2}{S_t=K }},
\end{equation}%
where $\eta$ denotes the local-volatility function, see~\cite{dupire1994pricing}. The familiar reader would recognize in Equation~\eqref{leverage computation} an application of the Markov-projection procedure, a concept originating from the celebrated Gy{\"o}ngy Lemma, see~\cite{gyongy1986mimicking}, whose idea basically lies in finding a diffusion which ``mimics'' the fixed-time marginal distributions of an It{\^o} process. 

For the derivation of Equation~\eqref{leverage computation}, we refer to~\cite{guyon2013nonlinear}. Notice that Equation~\eqref{leverage computation} is an implicit equation for $\ell$ as it is needed for the computation of $\mathbb{E}\left[\sigma_t^{2} \mid S_t=K\right]$. It is a standard thing in SLV models and this in turn means that the SDE for the price process $(S_t)_{t \geq 0}$ is actually a McKean-Vlasov SDE, since the probability distribution of $S_t$ enters the characteristics of the equation. Existence and uniqueness results for this equation are not obvious in principle, as the coefficients do not typically satisfy standard conditions like for instance Lipschitz continuity with respect to the so-called Wasserstein metric. In fact, deriving the set of stochastic volatility parameters for which SLV models exist uniquely, for a given market implied volatility surface, is a very challenging and still open problem. In the sequel, we will limit to investigate numerically the validity of the assumption just stated in our specific framework.

\section{Smile modelling for ETP on futures strategies}
\label{ETPstrategies}

We start this section with a wider look to study the volatility smile of a generic ETP based on futures strategies. Then, the results shown in this section will be applied to our specific VXX case in Section~\ref{An Example with VIX and VXX}.

\subsection{ETP on futures strategies}

We analyze an ETP based on a strategy of $N$ futures contracts $F_t^i$ with maturity dates $T_i$ and $i=1,\ldots,N$. If we assume that the ETP strategy itself is not collateralized and it requires a proportional fee payment $\phi_t$, we can write the strategy price process $V_t$ as given by
\begin{equation}
\label{VXX dynamics explicit}
    \frac{dV_t}{V_t} = (r_t - \phi_t) dt + \sum_{i=1}^N \omega_t^i \, \frac{dF_t^i}{F_t^i},
\end{equation}%
where $r_t$ is the risk-free rate. Moreover, we assume that the futures contracts and the ETP are expressed in the same currency as the bank account  based on $r_t$. The processes $\omega^i_t$ represent the investment percentage (or weights) in the futures contracts, and they are defined by the ETP term-sheet. Notice that the weights may depend on the whole vector of futures prices. It is straightforward to extend the analysis to more general ETP strategies based on total-return or dividend-paying indices. Moreover, we can introduce securities in foreign currencies as in~\cite{moreni2017derivative}.

We can consider as a first example the family of S\&P GSCI commodity indices. Exchange-Traded Commodity (ETC's) tracking these indices are liquid in the market, and they represent an investment in a specific commodity or commodity class. The investment is in the first nearby one-month future, but for the last five days of the contract, where a roll over procedure is implemented to sell the first nearby and to buy the second one. A second example are the volatility-based ETN's VXX and VXZ. They represent respectively an investment in short-term or in the long-term structure of VIX futures. The investment is in a portfolio of futures mimicking a rolling one-month contract. In Section~\ref{An Example with VIX and VXX} we will focus on the VXX case.

\subsection{Modeling the ETP smile}

We wish to jointly model the volatility smile of an ETP and of its underlying futures contracts. Indeed, market quotes for ETP plain-vanilla options can be found on the market. The main problem to solve is the fact that the price process of the ETP strategy is non-Markov, since it depends on the futures contract prices, namely the vector process $[V_t,F_t^1, \dots, F_t^N]$ is Markov. 

We start by introducing a generic stochastic volatility dynamics for futures prices under the risk-neutral measure, as given by
\begin{equation}
\label{SLV model for F}
    dF_t^i = \nu_t^i \cdot dW_t,
\end{equation}%
where $\nu_t^i$ are vector processes in $\mathbb{R}^d$, possibly depending on the futures prices themselves, and $W_t$ is a $d-$vector of standard Brownian motions under the risk-neutral measure. The notation $a \cdot b$ refers to the inner product between $a$ and $b$ (vectors or matrices).

In general, we can numerically solve together Equations~\eqref{VXX dynamics explicit} and~\eqref{SLV model for F} to calculate options prices on the ETP or on the futures contracts. Our aim is at calibrating both the futures and the ETP quoted smile. In the following, we always assume that interest rates and fees are deterministic functions of time, so that we write: $r_t:=r(t)$ and $\phi_t:=\phi(t)$. 

We can simplify the problem by noticing that we need only to know the marginal densities of the futures and ETP prices, so that we can focus our analysis on the Markov projection of the dynamics~\eqref{VXX dynamics explicit} and~\eqref{SLV model for F}. In this way, the model can  match the marginal densities of futures prices and ETP prices predicted by the local-volatility model.

\begin{Proposition}{\bf (Markov projections)}
\label{Prop: Markov projections}
    Given the dynamics~\eqref{SLV model for F} and~\eqref{VXX dynamics explicit}, the Markov projections of the processes $F_t^i$ are given by
    \begin{equation}
    \label{Futures Markov projection}
        d \widetilde{F}_t^i = \eta_F^i(t, \widetilde{F}_t^i) \,d\widetilde{W}_t^i
        \;,\quad
        \widetilde{F}_0^i = F_0^i
    \end{equation}%
    where $\widetilde{W}_t^i$ with $i=1, \dots, N$ are standard Brownian motions under the risk-neutral measure, $F_0^i$ are the futures prices as observed today in the market, and
    \begin{equation}
    \label{Futures Markov constraint}
        \eta_F^i(t, K) \doteq \sqrt{\ExC{0}{ \nu_t^i \cdot \nu_t^i}{F_t^i=K}},
    \end{equation}%
    while the Markov projections of the process $V_t$ is given by
    \begin{equation}
    \label{VXX Markov projection}
        \frac{d \widetilde{V}_t}{\widetilde{V}_t}= (r_t-\phi_t) \,dt + \eta_V(t, \widetilde{V}_t ) \,d\widetilde{W}_t^0
        \;,\quad
        \widetilde{F}_0^i = V_0
    \end{equation}%
    where $\widetilde{W}_t^0$ is a standard Brownian motion, $V_0^i$ is the ETP price as observed today in the market, and
    \begin{equation}
    \label{VXX Markov constraint}
        \eta_V(t,K) \doteq \sqrt{\ExC{0}{ \sum_{ij=1}^N \hat{\omega}_t^i \nu_t^i \cdot \nu_t^j \hat{\omega}_t^j }{V_t=K}},
    \end{equation}%
    with the weights $\hat{\omega}_t^i$ given by
    \begin{equation}
    \label{VXX tilde weights}
        \hat{\omega}_t^i \doteq \frac{\omega_t^i}{F_t^i}.
    \end{equation}%
\end{Proposition}

The main ingredient for proving Proposition~\ref{Prop: Markov projections} is the Gy{\"o}ngy Lemma, see~\cite{gyongy1986mimicking}, that for the sake of clarity we recall here below.

\begin{Lemma}{\bf (Gy{\"o}ngy Lemma)}
    \label{Gyongy Lemma}
        Let $X(t)$ be given by
    \begin{equation}
    \label{SDE general}
        X_t=X_0+\int_0^{t} \alpha(s) \, d s+\int_0^{t} \beta(s) \, d W_s, \quad t \geq 0
    \end{equation}%
    where $W$ is a Brownian motion under some probability measure $\mathbb{P}$, and $\alpha$, $\beta$ are adapted bounded stochastic processes such that~\eqref{SDE general} admits a unique solution. If we define $a(t,x)$, $b(t,x)$ by
    \begin{equation}
        \begin{aligned}
            a(t, x) &= \ExC{0}{ \alpha(t) }{ X(t)=x } \\
            b^{2}(t, x) &= \ExC{0}{ \beta^{2}(t) }{ X(t)=x },
        \end{aligned}
    \end{equation}%
    then there exists a filtered probability space $(\widetilde{\Omega}, \tilde{\mathcal{F}},\{ \tilde{\mathcal{F}} \}_{t \geq 0}, \widetilde{\mathbb{P}})$ where $\widetilde{W}$ is a $\widetilde{\mathbb{P}}-$Brownian motion, such that the SDE
    \begin{equation}
        \begin{aligned}
            Y_t=Y_0+\int_0^{t} a(s,Y(s)) \, d s+\int_0^{t} b(s,Y(s)) \, d \widetilde{W}_s, \quad t \geq 0
        \end{aligned}
    \end{equation}%
    admits a weak solution that has the same one-dimensional distributions as $X$. The process $Y$ is called the Markov projection of the process $X$.    
\end{Lemma}

\begin{proof} {\it (Proposition~\ref{Prop: Markov projections})~}
    Equation~\eqref{Futures Markov projection} follows from equation~\eqref{SLV model for F} and from Lemma~\ref{Gyongy Lemma}. Dynamics~\eqref{VXX dynamics explicit} can be rewritten as 
    \begin{equation}
        \frac{dV_t}{V_t} = (r_t - \phi_t) dt + \sum_{i=1}^{N} \omega_t^i \: \frac{\nu_t^i \cdot d W_t}{F_t^i},    
    \end{equation}%
    and applying Lemma \ref{Gyongy Lemma} we get 
    \begin{equation}
        \eta_V(t,K)^2 = \ExC{0}{ \sum_{i,j=1}^N \frac{\omega_t^i \nu_t^i}{F_t^i} \cdot \frac{\nu_t^j \omega_t^j}{F_t^j} }{ V_t=K },
    \end{equation}%
    while the drift part remains unaffected since both $r_t$ and $\phi_t$ are deterministic functions. The conclusion follows immediately.
\end{proof}

\subsection{Calibration of the local volatility functions}

The functions $\eta_F^i(t,K)$ and $\eta_V(t,K)$ are respectively the futures and ETP local-volatility functions. We can calibrate the model on plain-vanilla market quotes using the approach of~\cite{nastasi2020smile} originally implemented for the commodity market, that is $\eta_F^i(t,K)$ will be calibrated by using options on futures, while $\eta_V(t,K)$ by using options on the ETP strategy. In this approach the Markov projections $\widetilde{F}^i_t$ are modelled in term of a common process $s_t$, which we can identify with the price of a rolling futures contract.

We briefly describe the calibration procedure. We stress that the following assumptions on the dynamics of  $\widetilde{F}^i_t$ are used only to obtain a simple framework to calibrate the local volatility functions, and they do not pose any constraint on the futures price dynamics $F^i_t$.

We start by defining the futures local volatilities as given by
\begin{equation}
    \eta_F^i(t,K) \doteq \left( K - F_0^i \left( 1-e^{-\int_t^{T_i} a(u) \,du} \right) \right) \eta(t, k_F^i(t, K)), \label{local volatility function eta_F}
\end{equation}%
where $a$ is a non-negative function of time, $\eta$ is a positive and bounded function of time and price, while the effective strike $k_F^i(t,K)$ is given by
\begin{equation}
    k_F^i(t,K) \doteq 1 - \left( 1 - \frac{K}{F_0^i} \right) e^{\int_t^{T_i} a(u) \,du}.
\end{equation}%
Then, it is possible to show by a straightforward algebra that the above choice allows us to write
\begin{equation}
    \widetilde{F}^i_t = F_0^i \left( 1 - \left( 1 - s_t \right) e^{-\int_t^{T_i} a(u) \,du} \right)
\end{equation}%
where the driving risk process $s_t$ is modelled by a mean-reverting local-volatility dynamics, namely by
\begin{equation}
\label{Normalized spot price dynamics}
    ds_t = a(t) (1 - s_t) \,dt + \eta(t, s_t) s_t \,dW_t^s
    \;, \quad
    s_0 = 1,
\end{equation}%
where $W_t^s$ is a standard Brownian motion under the risk-neutral measure, and we get $W_t^s=\widetilde{W}_t^i$ for all the $i$.

Thus, plain-vanilla options on futures can be calculated as plain-vanilla options on the process $s_t$, which, in turn, can be easily computed by means of the Dupire equation. More precisely, given the dynamics in Equation~\eqref{Normalized spot price dynamics}, we have that the normalized call price 
\begin{equation}
    c(t, k)\doteq\mathbb{E}\left[\left(s_t-k\right)^{+}\right]
\end{equation}%
satisfies the parabolic PDE
\begin{equation}
\label{Extended Dupire equation}
    \partial_t c(t, k)=\left(-a(t)-a(t)(1-k) \partial_{k}+\frac{1}{2} k^{2} \eta^{2}(t, k) \partial_{k}^{2}\right) c(t, k),
\end{equation}%
with boundary conditions
\begin{equation}
\label{Extended Dupire equation Boundary conditions}
c(t, 0)=1, \quad c(t, \infty)=0, \quad c(0, k)=(1-k)^{+}.
\end{equation}%
A more detailed analysis and a proof of this (extended) Dupire equation can be found in~\citet[Proposition~4.1]{nastasi2020smile}. Given the prices of options on futures from the market, we can plug their corresponding normalized call prices in Equation~\eqref{Extended Dupire equation}, and then solve it  for the function $\eta$. Then, we can compute the futures local-volatility function $\eta_F(t,T_i,K)$, thanks to the mapping in Equation~\eqref{local volatility function eta_F},   and eventually recover the prices of options on futures.

There is a vast literature on how to solve Equation \eqref{Extended Dupire equation} with the boundary conditions~\eqref{Extended Dupire equation Boundary conditions}, here we use an implicit PDE discretization method, as usually done for solving the Dupire equation, see for example~\cite{gatheral2011volatility}. We allow a mean-reversion $a(t)$ possibly different from zero when we are calibrating the model to options on futures. On the other hand, options on the ETP spot price can be directly calibrated by a standard local-volatility model, that is without a mean-reversion term. Hence, once $a(t)$ is chosen, we can calibrate to the market $\eta_F^i(t,K)$ and $\eta_V(t,K)$, as previously described.

Thanks to the normalized spot dynamics~\eqref{Normalized spot price dynamics}, we can price all futures options by evaluating the PDE \eqref{Extended Dupire equation}, and ETP options as well. In this way we implicitly obtain the marginal probability densities of futures prices and ETP prices.

\subsection{Implementing the Markov-projection constraints}

We have now to ensure that the two constraints given by the Markov projection~\eqref{Futures Markov constraint} and~\eqref{VXX Markov constraint} hold for the futures price dynamics given by Equation~\eqref{SLV model for F}. More precisely, we need to define a suitable dynamics for the processes $\nu_t^i$ satisfying the constraints, which in turn enable us to preserve the calibration to plain-vanilla options. A possible strategy is given by selecting the vector processes $\nu_t^i$ in the following way.
\begin{equation}
\label{Process nu}
    \nu_t^i \doteq \ell_F^i(t, F_t^i) \, \sqrt{v_t} \, R_i(t, V_t),
\end{equation}%
where the volatility $v_t$ is a scalar positive process to be defined, $\ell^i_F(t,K)$ is the leverage function referred to the futures $F_t^i$, that can be computed as
\begin{equation}
\label{leverage function}
    \ell_F^i(t,K) \doteq \frac{\eta_F^i(t,K)}{\sqrt{\ExC{0}{ v_t }{F^i_t = K}}},
\end{equation}
and the $d$-dimensional vector function $R_i(t,V_t)$ plays the role of local correlation and satisfies the condition $\left\| R_i(t,V_t) \right\|=1$, where $\|\cdot\|$ is the Euclidean norm.

The case of $v_t$ given by a Cox-Ingersoll-Ross process corresponds to a Heston specification for the stochastic volatility. We adopt this choice in the following sections dedicated to the numerical investigations. Another natural candidate for $v_t$ is the Wishart process, which is well known to be able to describe a non trivial dependence among the (positive) diagonal factors, hence going beyond the classic multi-factor Heston model (see e.g. \cite{daf07,daf08}). In Appendix \ref{appendix} we provide a choice for the process $\nu_t^i$ that exploits the properties of the Wishart specification.

With the definition given by Equation~\eqref{Process nu} for the futures volatilities we obtain that the constraint~\eqref{Futures Markov constraint}   is satisfied by construction, so that the model reprices plain-vanilla options on futures correctly. Yet, the constraint~\eqref{VXX Markov constraint} is not automatically satisfied. If we substitute the definition of $\nu_t^i$ into Equation~\eqref{VXX Markov constraint}, we get
\begin{equation}
\label{VXX Markov Projection Solving Equation}
    \eta_V(t, K) = \sqrt{ \sum_{ij=1}^N R_i(t, K) \cdot R_j(t, K) \,\ExC{0}{ v_t \,\hat{\omega}_t^i \hat{\omega}_t^j \ell_F^i(t,F_t^i) \ell_F^j(t,F_t^j) }{ V_t=K } }.
\end{equation}%
Thus, in order to complete the calibration of our model, we need to solve the above equation for an (admissible) unknown local function $R_i(t,K)$. We will investigate numerically in practical cases the validity of this assumption.

\section{The VXX ETN strategy}
\label{An Example with VIX and VXX}

We consider the specific case of the VIX futures and the VXX ETN strategy. 

\subsection{Definition of the ETN strategy}

We start by labelling the futures maturity dates ${\cal T}:=\{T_1,\ldots,T_N\}$ from the first date after the observation date $t$ of the price of the ETN strategy. We define for $j=1,2,\ldots$
\begin{equation}
    T^{(0)} := \max\left\{T \in {\cal T} : T<t\right\}
    \;,\quad
    T^{(j)} := \min\left\{T \in {\cal T} : T>T_{j-1}\right\},
\end{equation}
where we assume that at least one maturity date in $\cal T$ is in the past of $t$. We name $F_t^{(j)}$ the futures contract with maturity $T^{(j)}$. In particular, we have that $F_t^{(1)}$ is the front contract, and $F_t^{(2)}$ is the second contract.

The VXX is defined as a strategy on the first and second nearby VIX futures as described in~\cite{gehricke2018modeling} and given by
\begin{equation}
    \frac{dV_t}{V_t} = (r_t - \phi_t) dt + \hat{\omega}_t^{(1)} \,dF_t^{(1)} + \hat{\omega}_t^{(2)} \,dF_t^{(2)},
\end{equation}%
where the weights are defined as
\begin{equation}
\label{VXX weights}
    \hat{\omega}_t^{(1)} \doteq \frac{\alpha(t)}{\alpha(t) F_t^{(1)}+(1-\alpha(t)) F_t^{(2)}}
    \;,\quad 
    \hat{\omega}_t^{(2)} \doteq \frac{1-\alpha(t)}{\alpha(t) F_t^{(1)}+(1-\alpha(t)) F_t^{(2)}}, 
\end{equation}%
where we define 
\begin{equation}
    \alpha(t) \doteq \frac{\varsigma(T^{(1)},-1) - \varsigma(t,1)}{\varsigma(T^{(1)},-1) - T^{(0)}},
\end{equation}%
and $\varsigma(t,d)$ is a calendar-date shift of $d$ business days.

We notice that the VXX strategy depends on the whole futures term structure, since the futures contracts we identify as front and second contracts change in time. On the other hand, at most two contracts enter the definition at each time, so that the strategy depends only on the correlations between the futures contracts $F^i_t$ and $F^{i+1}_t$ for all $i>0$.

\begin{Remark} {\bf (Definition of the ETP weights.)}
  In~\cite{grasselli2018vix} a slightly different definition is adopted, by using the year fraction between the second and the first nearby futures as reference time interval, instead of the year fraction between the first nearby and the previous futures contract (without calendar adjustments), leading to
  \begin{equation}
    \alpha(t) \doteq \frac{T^{(1)}-t}{T^{(2)}-T^{(1)}}.
  \end{equation}%
\end{Remark}

\subsection{The local correlation}

We now implement the modelling framework of the previous sections by using an explicit form for the local-correlation vector. We set the dimension of the Brownian motion vector $d=N$, where $N$ is the number of the futures contracts. Then, we define the local-correlation vector $R_i(t,K)$ as the $i$-th row of the Cholesky decomposition of the correlation matrix $C(t,K)$ whose elements are given by
\begin{equation}
\label{R vector parametrization}
    C_{ij}(t, K) \doteq \ind{i=j} + \ind{i\neq j} \rho(t, K),
\end{equation}%
where the scalar coefficient $\rho(t,K)$ has values in the range $[-1,1]$. The above definition satisfies $\left\| R_i(t,K) \right\| = 1$ for all $i$ by construction. Moreover, for any $i\neq j$ we get also $R_i(t,K) \cdot R_j(t,K) = \rho(t,K)$.

The particular expression we use for the local-correlation vector leads to simple calibration formulae since it allows us to explicitly calculate the coefficient $\rho(t,K)$ by solving Equation~\eqref{VXX Markov Projection Solving Equation}. In general, a more elaborated structure can be devised if more information on futures correlations are available in the market. For instance, we could consider the VXZ ETN to model more futures contracts, or we can resort to a historical analysis. We leave to a future work such analyses.

We can now calculate the coefficient $\rho(t,K)$ by solving Equation~\eqref{VXX Markov Projection Solving Equation}. We describe in the following proposition the results.

\begin{Proposition}
\label{Prop: rho computed VXX}
  Given the parametrization of Equation~\eqref{R vector parametrization}, the local correlation coefficient $\rho(t,K)$ is given by
  \begin{equation}
  \label{Rho equation}
      \rho(t,K) = \frac{\eta_V^2(t,K) - A^{(1)}(t, K)-A^{(2)}(t, K)}{2 A^{(12)}(t,K)},    
  \end{equation}%
  provided that $t$ is not on a futures maturity date, where 
  \begin{equation}
      A^{(1)}(t,K) \doteq \ExC{0}{ \left( \hat{\omega}_t^{(1)} \ell_F^{(1)}(t,F_t^{(1)}) \right)^2 v_t }{ V_t = K },
  \end{equation}%
  \begin{equation}
      A^{(2)}(t,K) \doteq \ExC{0}{ \left( \hat{\omega}_t^{(2)} \ell_F^{(2)}(t,F_t^{(2)}) \right)^2 v_t }{ V_t = K },
  \end{equation}%
  \begin{equation}
      A^{(12)}(t,K) \doteq \ExC{0}{ \hat{\omega}_t^{(1)} \hat{\omega}_t^{(2)} \ell_F^{(1)}(t,F_t^{(1)}) \ell_F^{(2)}(t,F_t^{(2)}) \, v_t }{ V_t = K }.
  \end{equation}%
\end{Proposition}%

\noindent In the proposition we name $\ell_F^{(1)}(t,K)$ the leverage function of the front futures contract, and accordingly the one of the second contract. The same we do in the proof where we name $\nu^{(1)}$ the volatility of front futures contract. Notice that when $t$ corresponds to a futures maturity date, Equation~\eqref{Rho equation} becomes meaningless, because the weights vanish.

\begin{proof}
    From the definition of the local-volatility $\eta_V(t,K)$ given by Equations~\eqref{VXX Markov constraint} and~\eqref{VXX Markov Projection Solving Equation}, we get
    \begin{equation}
    \label{Markov projection exploiting}
    \begin{split}
    \eta_V^2(t,K)
    &= \ExC{0}{ \left( \hat{\omega}_t^{(1)} \right)^2 \nu_t^{(1)} \cdot \nu_t^{(1)} + \left( \hat{\omega}_t^{(2)} \right)^2 \nu_t^{(2)} \cdot \nu_t^{(2)}  }{ V_t=K } \\
    &  \quad + 2 \,\ExC{0}{ \hat{\omega}_t^{(1)} \hat{\omega}_t^{(2)} \nu_t^{(1)} \cdot \nu_t^{(2)} }{ V_t=K } \\
    &= \ExC{0}{ \left( \hat{\omega}_t^{(1)} \ell_F^{(1)}(t,F_t^{(1)}) \right)^2 v_t }{ V_t = K } \\
    &  \quad + \ExC{0}{ \left( \hat{\omega}_t^{(2)} \ell_F^{(2)}(t,F_t^{(2)}) \right)^2 v_t }{ V_t = K } \\
    &  \quad + 2 \,\rho(t,K) \,\ExC{0}{ \hat{\omega}_t^{(1)} \hat{\omega}_t^{(2)} \ell_F^{(1)}(t,F_t^{(1)}) \ell_F^{(2)}(t,F_t^{(2)}) \,v_t }{ V_t = K }.
    \end{split}%
    \end{equation}%
    Therefore, when $t$ is not on a futures maturity date, we can obtain $\rho(t,K)$ from Equation~\eqref{Markov projection exploiting} and get the result.
\end{proof}

We can estimate the conditional expectations by means of the techniques developed in~\cite{guyon2012being}, namely approximating conditional expectations using smooth integration kernels. Thus, we can evaluate the expectation of a process $X_t$ given a second process $Y_t$ as
\begin{equation}
    \ExC{0}{ X_t }{ Y_t=K } \approx \frac{\Ex{0}{ X_t \delta^{\varepsilon}\left(Y_t-K\right) }}{\Ex{0}{ \delta^{\varepsilon}\left(Y_t-K\right) }},
\end{equation}%
where $\delta^{\varepsilon}$ is a suitably defined mollifier of the Dirac delta depending on a smoothing coefficient $\varepsilon$. We can use such approximation within the Monte Carlo simulation for $[V_t,F_t^1, \dots ,F_t^N]$ to evaluate the diffusion coefficients of the processes. Alternatively we could use the collocation method developed in~\cite{van2014heston}. We stress again that the existence of a solution for Equation~\eqref{Rho equation} with the above coefficients is an open question, and, moreover, even if the solution exists, we have to check against market quotes if the chosen parametrization is able to produce a local correlation $\rho(t,K) \in [-1,1]$ given the plain-vanilla option prices on futures contract and ETP strategy quoted in the market.

\begin{Remark} {\bf (Alternative correlation definition)}
In order to speed-up the calibration process we can use the following definition of the local-correlation vector exploiting the fact that the VXX ETN depends only on two futures contracts at a time. We set the dimension of the Brownian motion vector $d=2$, and we define the local-correlation vector $R_i(t,K)$ as given by
\begin{equation}
    R_i(t, K) \doteq \ind{i {\rm ~is~odd}} R^{(1)}(t,K) + \ind{i {\rm ~is~even}} R^{(2)}(t,K),
\end{equation}%
where the even and odd contributions are defined as
\begin{equation}
    R^{(1)}(t,K) \doteq \begin{bmatrix} 1 & 0 \end{bmatrix}
    \;,\quad
    R^{(2)}(t,K) \doteq \begin{bmatrix} \rho(t, K) & \sqrt{1 - \rho^2(t, K)} \end{bmatrix},
\end{equation}%
and the scalar coefficient $\rho(t,K)$ has values in the range $[-1,1]$. The above definition satisfies $\left\| R_i(t,K) \right\| = 1$ and $R_i(t,K) \cdot R_{i+1}(t,K) = \rho(t,K)$ for all $i$ by construction. We can select this alternative definition whenever we are only interested in the price of a contract depending on the front and second futures, as in the VXX case.
\end{Remark}

\subsection{Model specification}
\label{Model Specification}

We recall the stochastic local-volatility dynamics for the futures prices defined in Section~\ref{ETPstrategies}, and we select a specific form for the stochastic volatility. From Equations~\eqref{SLV model for F} and~\eqref{Process nu} we have under the risk-neutral measure
\begin{equation}
\label{SLV futures dynamics}
    dF_t^i = \ell_F^i(t,F_t^i) \sqrt{v_t} R_i(t,V_t) \cdot dW_t,
\end{equation}%
where the local-correlation vector $R_i(t,V_t)$ is given by Equatios~\ref{R vector parametrization}. We assume that the process $v_t$ follows a one-dimensional CIR process, namely
\begin{equation}
    dv_t = \kappa (\theta-v_t) \,dt + \xi \sqrt{v_t} \,dZ_t,
\end{equation}%
with initial condition $v_0 = \bar{v}$. The Brownian motion $Z_t$ is correlated with the component of the vector $W_t$. In particular, we select the same correlation for each component, namely $d\langle W^i_t, Z_t \rangle_t = \rho_v \,dt$ with $\rho_v \in [-1,1]$.

In the above equation the positive constants $\kappa$, $\theta$ and $\xi$ are respectively the speed of the mean-reversion of the variance, the long-run average variance, and the volatility-of-volatility parameter. Our one-factor specification for the stochastic volatility will turn out to be enough for our purposes. In fact, the apparently more general multifactor choice does not lead to better results in the quality of the fit, up to our numerical experiences.

Finally, we assume also that the mean reversion $a(t)$ is independent of time to reduce the number of free parameters. Thus, we can collect all the parameters of our model not fixed by the calibration procedure in the following vector
\begin{equation}
\label{Hyper}
    \Theta := \left\{ a, \bar{v}, \kappa, \theta, \xi, \rho_v \right\}.
\end{equation}
These parameters can be fine-tuned so that Equation~\eqref{Rho equation} leads to a local-correlation coefficient in the valid range $[-1,1]$. We will see in the next section how to achieve this goal.

\section{Numerical illustration on a real data set}
\label{Numerical Results}

In this section we perform a numerical exercise, based on real data quoted by the market for options on the VIX futures and the VXX ETN strategy, where we provide a set of admissible parameters, that is leading to an admissible local correlation vector in Equation~\eqref{Rho equation} satisfying $\rho(t,K) \in [-1,1]$. In this way we build up a coherent framework that correctly reprices plain vanilla options on futures contract and VXX strategy quoted in the market. We stress that the set of parameters we are going to find does not necessarily represent the optimal choice, namely we are just  providing an admissible framework, not a full calibration. We leave for future research the complete exploration of this issue.

\subsection{VIX and VXX data sets}

\begin{table}
  \centering
  \begin{tabular}{|c|c|c|c|}
    \hline
    \multicolumn{4}{|c|}{\rule{0pt}{2.5ex} VIX Data Set} \\
    \hline \rule{0pt}{2.5ex} $T^{VIX}$ & $F_0(T)$ & \# options & $K$ (min $-$ max) \\
    \hline
    \multicolumn{1}{|l|}{$T_1$: November $20$, $2019$} & 14.60 & $40$ & $10-80$ \\
    \multicolumn{1}{|l|}{$T_2$: December $18$, $2019$} & 16.15 & $40$ & $10-80$ \\
    \multicolumn{1}{|l|}{$T_3$: January $22$, $2020$} & 17.45 & $35$ & $10-80$ \\
    \multicolumn{1}{|l|}{$T_4$: February $19$, $2020$} & 18.15 & $35$ & $10-80$ \\
    \hline
  \end{tabular}
  \caption{VIX futures call options data set as on $T_0$: November 7, 2019. } 
\label{tab: VIX futures call options data set}
\end{table}

\begin{table}
  \centering
  \begin{tabular}{|c|c|c|}
    \hline
    \multicolumn{3}{|c|}{\rule{0pt}{2.5ex} VXX Data Set} \\
    \hline \rule{0pt}{2.5ex} $T^{VXX}$ & \# options & $K$ (min $-$ max) \\
    \hline
    \multicolumn{1}{|l|}{$T_1$: November 20, 2019} & $65$ & $8-60$ \\
    \multicolumn{1}{|l|}{$T_2$: December 18, 2019} & $59$ & $2-60$ \\
    \multicolumn{1}{|l|}{$T_3$: January 22, 2020} & $53$ & $10-90$ \\
    \hline
  \end{tabular}
  \caption{VXX call options data set as on $T_0$: November 7, 2019. Spot price: $V_0 = 19.22$.} 
\label{tab: VXX call options data set}
\end{table}

We test our model against a set of three maturities for VXX call options with a corresponding set of four maturities for VIX futures call options. All data are observed on the market on November 7, 2019. Table~\ref{tab: VIX futures call options data set} (resp.\ Table~\ref{tab: VXX call options data set}) describes the features of the VIX (resp.\ VXX) options data set. In the tables we list also the different maturities of VIX and VXX options involved in the numerical exercise, and the values of VIX futures and VXX spot price, as observed in the market. Market data quotes are shown in the calibration section.

\begin{figure}
\centering
\scalebox{1.2}{
\begin{tikzpicture}
  \draw[->,line width=2pt] (0,0) -- (11.5,0);
  \draw[->,line width=2pt] (0,-1.5) -- (11.5,-1.5);
  \draw[dotted,line width=1pt,color=black!75] (5.4,-2) -- (5.4,2) node [above,rotate=90,xshift=-15pt,yshift=0pt] {{\footnotesize \scalebox{.7}[1.0]{2019}}};
  \path[color=black!75] (5.4,-2) -- (5.4,2) node [below,rotate=90,xshift=-15pt,yshift=0pt] {{\footnotesize \scalebox{.7}[1.0]{2020}}};

  \draw (0.2,0) node [below] {\rule{0pt}{2.5ex}{\footnotesize$\;T_0$}} -- (0.2,-0.1) node [right,rotate=90,xshift=5pt,yshift=0pt] {{\footnotesize \scalebox{.7}[1.0]{Nov 7}}};
  \draw (1.5,0) node [below] {\rule{0pt}{2.5ex}{\footnotesize$T^{VIX}_1$}} -- (1.5,-0.1) node [right,rotate=90,xshift=5pt,yshift=0pt] {{\footnotesize \scalebox{.7}[1.0]{Nov 20}}};
  \draw (4.1,0) node [below] {\rule{0pt}{2.5ex}{\footnotesize$\!\!\!\!\!\!\!\!\!T^{VIX}_2$}} -- (4.1,-0.1) node [right,rotate=90,xshift=5pt,yshift=3pt] {{\footnotesize \scalebox{.7}[1.0]{Dec 18}}};
  \draw (7.8,0) node [below] {\rule{0pt}{2.5ex}{\footnotesize$\;T^{VIX}_3$}} -- (7.8,-0.1) node [right,rotate=90,xshift=5pt,yshift=0pt] {{\footnotesize \scalebox{.7}[1.0]{Jan 22}}};
  \draw (10.4,0) node [below] {\rule{0pt}{2.5ex}{\footnotesize$T^{VIX}_4$}} -- (10.4,-0.1) node [right,rotate=90,xshift=5pt,yshift=0pt] {{\footnotesize \scalebox{.7}[1.0]{Feb 19}}};

  \draw (1.0,0) node [right,rotate=90,xshift=1.9pt,yshift=0pt] {{\footnotesize \scalebox{.7}[1.0]{Nov 15}}} -- (1.0,-1.6) node [below] {\rule{0pt}{2.3ex}{\footnotesize$T^{VXX}_1$}};
  \draw (4.3,0) node [right,rotate=90,xshift=1.9pt,yshift=-3pt] {{\footnotesize \scalebox{.7}[1.0]{Dec 20}}} -- (4.3,-1.6) node [below] {\rule{0pt}{2.3ex}{\footnotesize$T^{VXX}_2$}};
  \draw (7.3,0) node [right,rotate=90,xshift=1.9pt,yshift=0pt] {{\footnotesize \scalebox{.7}[1.0]{Jan 17}}} -- (7.3,-1.6) node [below] {\rule{0pt}{2.3ex}{\footnotesize$T^{VXX}_3$}};
\end{tikzpicture}}%
\caption{Term structure of VIX and VXX option maturities starting from November $7$, $2019$.}
\label{fig:VIX and VXX maturities}
\end{figure}
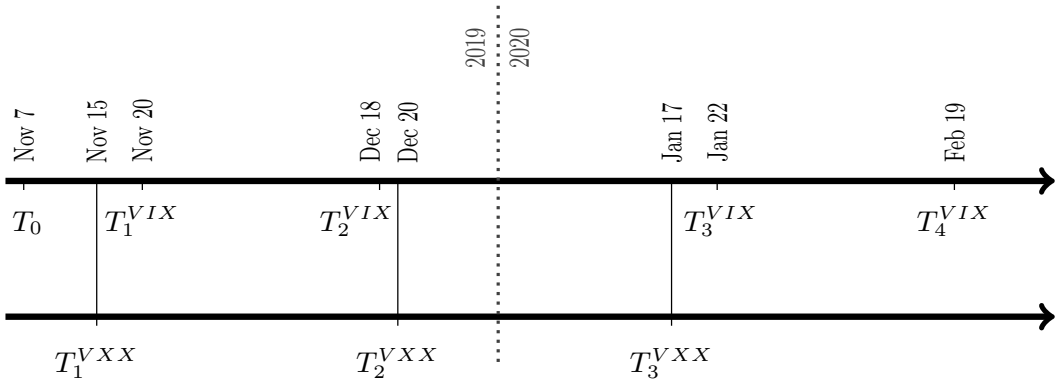

\subsection{Parameter sensitivities}
\label{subsection: Parameters Sensitivity}

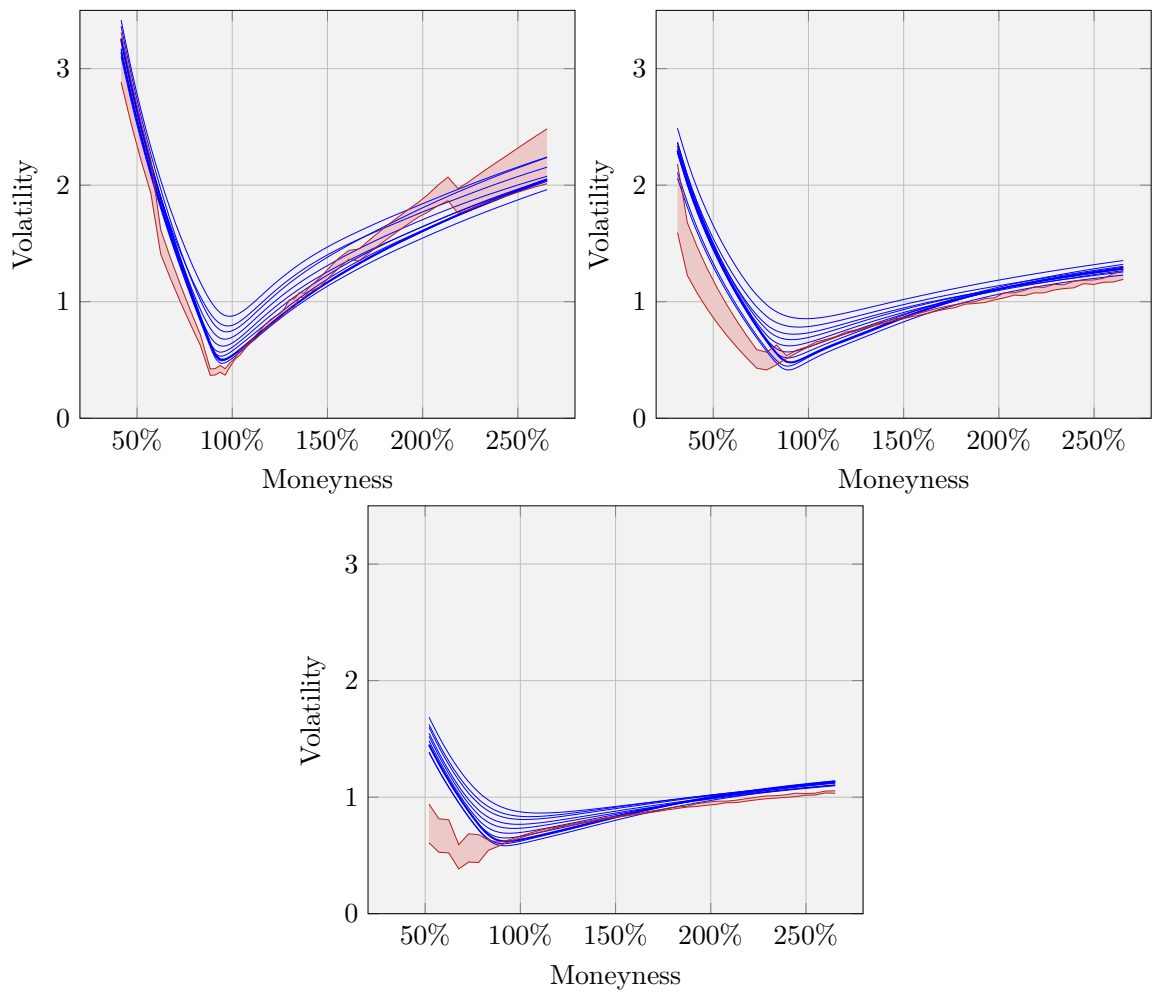
\begin{figure}
    \begin{center}
    \scalebox{0.95}{%
    \begin{tikzpicture}
    \begin{axis}[no markers,
                xlabel=Moneyness,
                ylabel=Volatility,
                ylabel style={overlay},
                restrict x to domain=0.3:2.7,
                xmin=0.2, xmax=2.8,
                xtick={0.5,1,1.5,2,2.5},
                xticklabels={50\%,100\%,150\%,200\%,250\%},
                ymin=0, ymax=3.5,
                grid=major,
                axis background/.style={fill=gray!10}]
    \addplot [color=blue,ultra thin,smooth] table [y=vxxLVam0,x expr=\thisrow{vxxKa}/19.22] from \vxxmean;
    \addplot [color=blue,ultra thin,smooth] table [y=vxxLVam1,x expr=\thisrow{vxxKa}/19.22] from \vxxmean;
    \addplot [color=blue,ultra thin,smooth] table [y=vxxLVam2,x expr=\thisrow{vxxKa}/19.22] from \vxxmean;
    \addplot [color=blue,ultra thin,smooth] table [y=vxxLVam3,x expr=\thisrow{vxxKa}/19.22] from \vxxmean;
    \addplot [color=blue,ultra thin,smooth] table [y=vxxLVam4,x expr=\thisrow{vxxKa}/19.22] from \vxxmean;
    \addplot [color=blue,ultra thin,smooth] table [y=vxxLVam5,x expr=\thisrow{vxxKa}/19.22] from \vxxmean;
    \addplot [color=blue,ultra thin,smooth] table [y=vxxLVam6,x expr=\thisrow{vxxKa}/19.22] from \vxxmean;
    \addplot [color=blue,ultra thin,smooth] table [y=vxxLVam7,x expr=\thisrow{vxxKa}/19.22] from \vxxmean;
    \addplot [color=blue,ultra thin,smooth] table [y=vxxLVam8,x expr=\thisrow{vxxKa}/19.22] from \vxxmean;
    \addplot [color=blue,thick,smooth] table [y=vxxLVam75,x expr=\thisrow{vxxKa}/19.22] from \vxxmean;
    \addplot [name path=upper,color=mydarkred,thin] table [y expr=\thisrow{vxxMKTa}+\thisrow{vxxBAa}/2,x expr=\thisrow{vxxKa}/19.22] from \vixvxx;
    \addplot [name path=lower,color=mydarkred,thin] table [y expr=\thisrow{vxxMKTa}-\thisrow{vxxBAa}/2,x expr=\thisrow{vxxKa}/19.22] from \vixvxx;
    \addplot [fill=mydarkred!25] fill between[of=upper and lower];  
    \end{axis}
    \end{tikzpicture}}%
    \hspace*{0.5cm}
    \scalebox{0.95}{%
    \begin{tikzpicture}
    \begin{axis}[no markers,
                xlabel=Moneyness,
                ylabel=Volatility,
                ylabel style={overlay},
                restrict x to domain=0.3:2.7,
                xmin=0.2, xmax=2.8,
                xtick={0.5,1,1.5,2,2.5},
                xticklabels={50\%,100\%,150\%,200\%,250\%},
                ymin=0, ymax=3.5,
                grid=major,
                axis background/.style={fill=gray!10}]
    \addplot [color=blue,ultra thin,smooth] table [y=vxxLVbm0,x expr=\thisrow{vxxKb}/19.22] from \vxxmean;
    \addplot [color=blue,ultra thin,smooth] table [y=vxxLVbm1,x expr=\thisrow{vxxKb}/19.22] from \vxxmean;
    \addplot [color=blue,ultra thin,smooth] table [y=vxxLVbm2,x expr=\thisrow{vxxKb}/19.22] from \vxxmean;
    \addplot [color=blue,ultra thin,smooth] table [y=vxxLVbm3,x expr=\thisrow{vxxKb}/19.22] from \vxxmean;
    \addplot [color=blue,ultra thin,smooth] table [y=vxxLVbm4,x expr=\thisrow{vxxKb}/19.22] from \vxxmean;
    \addplot [color=blue,ultra thin,smooth] table [y=vxxLVbm5,x expr=\thisrow{vxxKb}/19.22] from \vxxmean;
    \addplot [color=blue,ultra thin,smooth] table [y=vxxLVbm6,x expr=\thisrow{vxxKb}/19.22] from \vxxmean;
    \addplot [color=blue,ultra thin,smooth] table [y=vxxLVbm7,x expr=\thisrow{vxxKb}/19.22] from \vxxmean;
    \addplot [color=blue,ultra thin,smooth] table [y=vxxLVbm8,x expr=\thisrow{vxxKb}/19.22] from \vxxmean;
    \addplot [color=blue,thick,smooth] table [y=vxxLVbm75,x expr=\thisrow{vxxKb}/19.22] from \vxxmean;
    \addplot [name path=upper,color=mydarkred,thin] table [y expr=\thisrow{vxxMKTb}+\thisrow{vxxBAb}/2,x expr=\thisrow{vxxKb}/19.22] from \vixvxx;
    \addplot [name path=lower,color=mydarkred,thin] table [y expr=\thisrow{vxxMKTb}-\thisrow{vxxBAb}/2,x expr=\thisrow{vxxKb}/19.22] from \vixvxx;
    \addplot [fill=mydarkred!25] fill between[of=upper and lower];  
    \end{axis}
    \end{tikzpicture}}\\
    \scalebox{0.95}{%
    \begin{tikzpicture}
    \begin{axis}[no markers,
                xlabel=Moneyness,
                ylabel=Volatility,
                ylabel style={overlay},
                restrict x to domain=0.3:2.7,
                xmin=0.2, xmax=2.8,
                xtick={0.5,1,1.5,2,2.5},
                xticklabels={50\%,100\%,150\%,200\%,250\%},
                ymin=0, ymax=3.5,
                grid=major,
                axis background/.style={fill=gray!10}]
    \addplot [color=blue,ultra thin,smooth] table [y=vxxLVcm0,x expr=\thisrow{vxxKc}/19.22] from \vxxmean;
    \addplot [color=blue,ultra thin,smooth] table [y=vxxLVcm1,x expr=\thisrow{vxxKc}/19.22] from \vxxmean;
    \addplot [color=blue,ultra thin,smooth] table [y=vxxLVcm2,x expr=\thisrow{vxxKc}/19.22] from \vxxmean;
    \addplot [color=blue,ultra thin,smooth] table [y=vxxLVcm3,x expr=\thisrow{vxxKc}/19.22] from \vxxmean;
    \addplot [color=blue,ultra thin,smooth] table [y=vxxLVcm4,x expr=\thisrow{vxxKc}/19.22] from \vxxmean;
    \addplot [color=blue,ultra thin,smooth] table [y=vxxLVcm5,x expr=\thisrow{vxxKc}/19.22] from \vxxmean;
    \addplot [color=blue,ultra thin,smooth] table [y=vxxLVcm6,x expr=\thisrow{vxxKc}/19.22] from \vxxmean;
    \addplot [color=blue,ultra thin,smooth] table [y=vxxLVcm7,x expr=\thisrow{vxxKc}/19.22] from \vxxmean;
    \addplot [color=blue,ultra thin,smooth] table [y=vxxLVcm8,x expr=\thisrow{vxxKc}/19.22] from \vxxmean;
    \addplot [color=blue,thick,smooth] table [y=vxxLVcm75,x expr=\thisrow{vxxKc}/19.22] from \vxxmean;
    \addplot [name path=upper,color=mydarkred,thin] table [y expr=\thisrow{vxxMKTc}+\thisrow{vxxBAc}/2,x expr=\thisrow{vxxKc}/19.22] from \vixvxx;
    \addplot [name path=lower,color=mydarkred,thin] table [y expr=\thisrow{vxxMKTc}-\thisrow{vxxBAc}/2,x expr=\thisrow{vxxKc}/19.22] from \vixvxx;
    \addplot [fill=mydarkred!25] fill between[of=upper and lower];  
    \end{axis}
    \end{tikzpicture}}
    \end{center}
    \caption{Impact on the VXX smile of the mean-reversion speed. Market bid-ask quotes span the red shaded area, while model implied volatilities are the blue solid lines with mean-reversion speed ranging from $0$ to $8$. The thick blue line is $a=7.5$.}
\label{fig:impact of mean reversion}
\end{figure}

\begin{figure}
    \begin{center}
    \scalebox{0.95}{%
    \begin{tikzpicture}
    \begin{axis}[no markers,
                xlabel=Moneyness,
                ylabel=Volatility,
                ylabel style={overlay},
                restrict x to domain=0.3:2.7,
                xmin=0.2, xmax=2.8,
                xtick={0.5,1,1.5,2,2.5},
                xticklabels={50\%,100\%,150\%,200\%,250\%},
                ymin=0, ymax=3.5,
                grid=major,
                axis background/.style={fill=gray!10}]
    \addplot [color=blue,ultra thin,smooth] table [y=vxxLVax01,x expr=\thisrow{vxxKa}/19.22] from \vxxvolvol;
    \addplot [color=blue,ultra thin,smooth] table [y=vxxLVax04,x expr=\thisrow{vxxKa}/19.22] from \vxxvolvol;
    \addplot [color=blue,ultra thin,smooth] table [y=vxxLVax07,x expr=\thisrow{vxxKa}/19.22] from \vxxvolvol;
    \addplot [color=blue,ultra thin,smooth] table [y=vxxLVax10,x expr=\thisrow{vxxKa}/19.22] from \vxxvolvol;
    \addplot [color=blue,ultra thin,smooth] table [y=vxxLVax13,x expr=\thisrow{vxxKa}/19.22] from \vxxvolvol;
    \addplot [color=blue,thick,smooth] table [y=vxxLVax11,x expr=\thisrow{vxxKa}/19.22] from \vxxvolvol;
    \addplot [name path=upper,color=mydarkred,thin] table [y expr=\thisrow{vxxMKTa}+\thisrow{vxxBAa}/2,x expr=\thisrow{vxxKa}/19.22] from \vixvxx;
    \addplot [name path=lower,color=mydarkred,thin] table [y expr=\thisrow{vxxMKTa}-\thisrow{vxxBAa}/2,x expr=\thisrow{vxxKa}/19.22] from \vixvxx;
    \addplot [fill=mydarkred!25] fill between[of=upper and lower];  
    \end{axis}
    \end{tikzpicture}}%
    \hspace*{0.5cm}
    \scalebox{0.95}{%
    \begin{tikzpicture}
    \begin{axis}[no markers,
                xlabel=Moneyness,
                ylabel=Volatility,
                ylabel style={overlay},
                restrict x to domain=0.3:2.7,
                xmin=0.2, xmax=2.8,
                xtick={0.5,1,1.5,2,2.5},
                xticklabels={50\%,100\%,150\%,200\%,250\%},
                ymin=0, ymax=3.5,
                grid=major,
                axis background/.style={fill=gray!10}]
    \addplot [color=blue,ultra thin,smooth] table [y=vxxLVbx01,x expr=\thisrow{vxxKb}/19.22] from \vxxvolvol;
    \addplot [color=blue,ultra thin,smooth] table [y=vxxLVbx04,x expr=\thisrow{vxxKb}/19.22] from \vxxvolvol;
    \addplot [color=blue,ultra thin,smooth] table [y=vxxLVbx07,x expr=\thisrow{vxxKb}/19.22] from \vxxvolvol;
    \addplot [color=blue,ultra thin,smooth] table [y=vxxLVbx10,x expr=\thisrow{vxxKb}/19.22] from \vxxvolvol;
    \addplot [color=blue,ultra thin,smooth] table [y=vxxLVbx11,x expr=\thisrow{vxxKb}/19.22] from \vxxvolvol;
    \addplot [color=blue,thick,smooth] table [y=vxxLVbx10,x expr=\thisrow{vxxKb}/19.22] from \vxxvolvol;
    \addplot [name path=upper,color=mydarkred,thin] table [y expr=\thisrow{vxxMKTb}+\thisrow{vxxBAb}/2,x expr=\thisrow{vxxKb}/19.22] from \vixvxx;
    \addplot [name path=lower,color=mydarkred,thin] table [y expr=\thisrow{vxxMKTb}-\thisrow{vxxBAb}/2,x expr=\thisrow{vxxKb}/19.22] from \vixvxx;
    \addplot [fill=mydarkred!25] fill between[of=upper and lower];  
    \end{axis}
    \end{tikzpicture}}\\
    \scalebox{0.95}{%
    \begin{tikzpicture}
    \begin{axis}[no markers,
                xlabel=Moneyness,
                ylabel=Volatility,
                ylabel style={overlay},
                restrict x to domain=0.3:2.7,
                xmin=0.2, xmax=2.8,
                xtick={0.5,1,1.5,2,2.5},
                xticklabels={50\%,100\%,150\%,200\%,250\%},
                ymin=0, ymax=3.5,
                grid=major,
                axis background/.style={fill=gray!10}]
    \addplot [color=blue,ultra thin,smooth] table [y=vxxLVcx01,x expr=\thisrow{vxxKc}/19.22] from \vxxvolvol;
    \addplot [color=blue,ultra thin,smooth] table [y=vxxLVcx04,x expr=\thisrow{vxxKc}/19.22] from \vxxvolvol;
    \addplot [color=blue,ultra thin,smooth] table [y=vxxLVcx07,x expr=\thisrow{vxxKc}/19.22] from \vxxvolvol;
    \addplot [color=blue,ultra thin,smooth] table [y=vxxLVcx10,x expr=\thisrow{vxxKc}/19.22] from \vxxvolvol;
    \addplot [color=blue,ultra thin,smooth] table [y=vxxLVcx13,x expr=\thisrow{vxxKc}/19.22] from \vxxvolvol;
    \addplot [color=blue,thick,smooth] table [y=vxxLVcx11,x expr=\thisrow{vxxKc}/19.22] from \vxxvolvol;
    \addplot [name path=upper,color=mydarkred,thin] table [y expr=\thisrow{vxxMKTc}+\thisrow{vxxBAc}/2,x expr=\thisrow{vxxKc}/19.22] from \vixvxx;
    \addplot [name path=lower,color=mydarkred,thin] table [y expr=\thisrow{vxxMKTc}-\thisrow{vxxBAc}/2,x expr=\thisrow{vxxKc}/19.22] from \vixvxx;
    \addplot [fill=mydarkred!25] fill between[of=upper and lower];  
    \end{axis}
    \end{tikzpicture}}
    \end{center}
    \caption{Impact on the VXX smile of the volatility of volatility. Market bid-ask quotes span the red shaded area, while model implied volatilities are the blue solid lines with volatilty of volatility ranging from $0.1$ to $1.3$. The thick blue line is $\xi=1.1$.}
\label{fig:impact of volvol}
\end{figure}
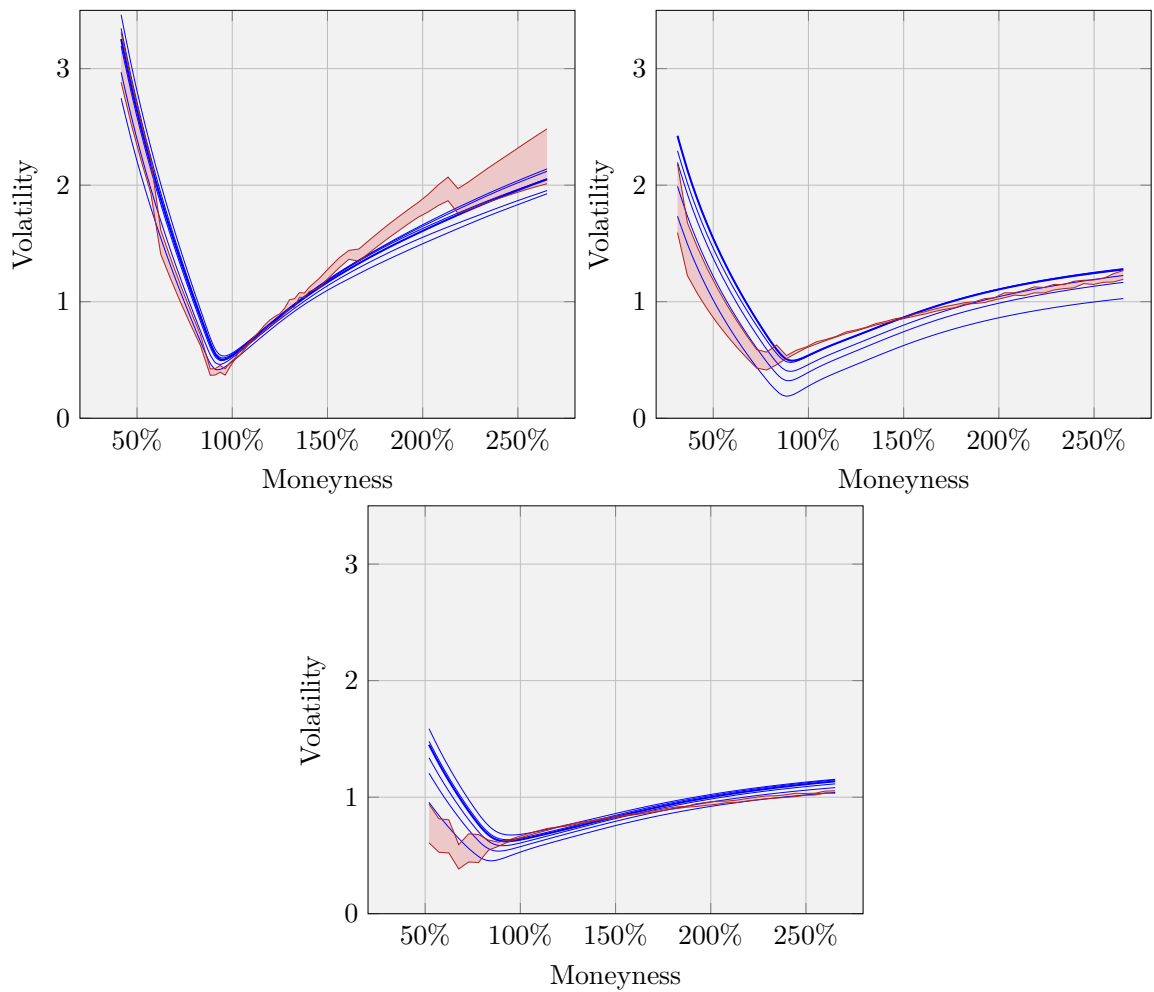

We now perform a sensitivity analysis that will highlight the impact of the model parameters on the VXX smiles. The aim of this analysis is to find a range of parameters that guarantees a good fit to the VXX market smile, while maintaining the fit on the VIX market smile. Then we test if these parameters  generate a local correlation $\rho(t,K) \in [-1,1]$, by solving Equation~\eqref{Rho equation}. More precisely, we let each parameter vary while keeping the others frozen. In this way, we can check the overall impact of each parameter on the shape of the VXX smile. This will give us an insight on  the set of reasonable  parameters that can generate a good fit, provided that the local correlation in Equation~\eqref{Rho equation} is admissible. The analysis is performed by fixing the correlation $\rho=0.85$.

We shall investigate the impact of the mean-reversion speed $a$ and the volatility-of-volatility $\xi$. We omit an explicit analysis of the impact of the other parameters $\bar{v}$, $\kappa$, $\theta$ and $\rho^v$, since it turns out empirically that they do not play an important role in the smile modelling. We set in the following $\rho^v=0.75$, and we choose the remaining three parameters so that the Feller condition holds. In particular, we choose $\kappa = 2.5$, $\theta = 2.5$, $\bar{v} = 1$. 

We start by the impact of the mean-reversion $a$. We perform the simulations with volatility-of-volatility $\xi=1.1$, while the mean-reversion speed can range from $0$ to $8$. In Figure~\ref{fig:impact of mean reversion} we present the results. We get that the higher the mean-reversion, the lower the smile. Thus, in order reproduce the market ATM implied volatility level for the VXX smile we need a high mean-reversion speed, the best match being $a=7.5$, in particular if we look near the at-the-money level of the first maturity. The (high) value for the mean-reversion speed could seem misconceiving. However, we should  keep in mind that we are jointly modelling two markets, the VIX and VXX option books, that display different peculiarities, therefore a somehow non standard level for the mean reversion seems a reasonable price to pay in order to reach our goal.

Then, we consider the impact of the volatility-of-volatility $\xi$. We perform the simulations with mean-reversion speed $a=7.5$, while the volatility of volatility can range from $0.1$ to $1.3$. In Figure~\ref{fig:impact of volvol} we present the results. We observe a relevant impact on the left wing of the smile. The best match seems to be for $\xi=1.1$, in particular if we look near the at-the-money level of the first maturity.

In the next subsection, we are going to build up a consistent framework in order to fix an admissible choice for all parameters that allows us to reproduce the market smiles of both VIX and VXX. In particular we check a posteriori that the choices of the mean-reversion speed and of the volatility of volatility are compatible with an admissible local correlation for most of the simulation paths.

\subsection{Joint fit of VIX and VXX market smiles}

\begin{table}
  \centering
  \begin{tabular}{|c|c|c|c|c|}
    \hline \rule{0pt}{2.5ex} $T^{VIX}$ & $K$ & Ask & Bid & Model \\
    \hline
    & $14.0$ & 1.2346 & 0.9477 & 1.0739\\
    November 20, 2019 & $14.5$ & 1.2079 & 0.9793 & 1.0911 \\
    & $15.0$ & 1.2702 & 0.9793 & 1.1392 \\
    \hline
    & $15.0$ & 0.9409 & 0.8149 & 0.8642 \\
    December 18, 2019 & $16.0$ & 0.9583 & 0.8407 & 0.9030 \\
    & $17.0$ & 0.9932 & 0.8880 & 0.9170 \\
    \hline
    & $16.0$ & 0.7818 & 0.6776 & 0.7263 \\
    January 22, 2020 & $17.0$ & 0.8107 & 0.7132 & 0.7540 \\
    & $18.0$ & 0.8336 & 0.7386 & 0.7839 \\
    \hline
    & $17.0$ & 0.7478 & 0.6506 & 0.6896 \\
    February 19, 2020 & $18.0$ & 0.7555 & 0.6759 & 0.7103 \\
    & $19.0$ & 0.7881 & 0.6969 & 0.7360 \\
    \hline
  \end{tabular}
  \caption{Implied volatility comparison between market quotes as on November 7, 2019, and the implied volatility obtained via our SLV models for VIX call options, for the four maturities and around ATM.}
\label{tab: VIX ATM table}
\end{table}

\begin{table}
  \centering
  \begin{tabular}{|c|c|c|c|c|}
    \hline \rule{0pt}{2.5ex} $T^{VXX}$ & $K$ & Ask & Bid & Model \\
    \hline
    & $19$ & 0.4776 & 0.4428 & 0.4602 \\
    November 15, 2019 & $19.5$ & 0.5153 & 0.5019 & 0.5030 \\
    & $20$ & 0.5844 & 0.5426 & 0.5473 \\
    \hline
    & $19$ & 0.6116 & 0.6040 & 0.6050 \\
    December 20, 2019 & $19.5$ & 0.6333 & 0.6101 & 0.6154 \\
    & $20$ & 0.6549 & 0.6311 & 0.6329 \\
    \hline
    & $19$ & 0.6603 & 0.6443 & 0.6450 \\
    January 17, 2020 & $19.5$ & 0.6736 & 0.6503 & 0.6514 \\
    & $20$ & 0.6870 & 0.6723 & 0.6735 \\
    \hline
  \end{tabular}
  \caption{Implied volatility comparison between market quotes as on November 7, 2019 and the implied volatility obtained via our SLV models for VXX call options, for the three maturities and around ATM.}
\label{tab: VXX ATM table}
\end{table}

Thanks to the analysis performed in Subsection~\ref{subsection: Parameters Sensitivity}, we are now able to pick a mix of parameters that enables us to get a good fit to both the VIX and VXX market smiles and that leads  to an admissible correlation parameter $\rho(t,K) \in [-1,1]$. That is, the model gives the proper correlation we need to put in the VIX and VXX dynamics, in order to reprice correctly the corresponding options, as described in Section~\ref{An Example with VIX and VXX}. Let us summarize the parameters that we have fixed in order to perform the joint fit.

\begin{itemize}
\item The mean-reversion speed used in the joint fit is $a = 7.5$ (this value revealed to be the optimal trade-off in order to recover the ATM implied-volatility level for the VXX options, while keeping a good fit for the VIX market).
\item The volatility-of-volatility parameter is fixed as  $\xi = 1.1$.
\item The correlation between futures and their volatilities is fixed as $\rho^v = 0.75$.
\item The remaining parameters used in the CIR-like volatility process are chosen by taking into account the non-violation of the Feller condition. In particular, we choose $\kappa = 2.5$, $\theta = 2.5$, $\bar{v} = 1$.
\end{itemize}

We have now all the ingredients to perform the joint fit of the implied volatilities of call options on VIX and VXX for different strikes and maturities around the ATM via the SLV method described previously which employs a Monte Carlo simulation with $5 \cdot 10^5$ paths. In Table~\ref{tab: VIX ATM table} (resp.\ Table~\ref{tab: VXX ATM table}) we show the model implied volatilities on VIX (resp.\ VXX) around ATM, that are generated by the SLV model,  and the market ones, with the relative errors.

In Figure~\ref{fig: VIX model} (resp.\ Figure~\ref{fig: VXX model}) we show the whole implied volatilities for all strikes  obtained by our SLV model  and the market volatilities of VIX futures (resp.\ VXX) call options. Despite the irregular shape of the market VXX smiles, the fit is overall good and even remarkably good for the VIX smiles.

We complete the discussion of the results by showing some histograms of the distribution of the local correlation parameter $\rho(t,V_t)$, when evaluated on the VXX value simulated along each Monte Carlo path. We are able to check that for most of the paths the correlation values are within the admissible interval $ [-1,1]$, as required in Proposition~\ref{Prop: rho computed VXX}. In Figure~\ref{fig:histo} we show the histograms for different times. In particular we notice that the distribution of the local correlation quickly flattens as the time elapses and it stabilizes around a mean value of 22\%, with a standard deviation of 32\%, and with less than the 3\% of paths with a local correlation value exceeding $1$. The local correlation in the simulation is capped to the maximum allowed value of $1$. The satisfactory fit to VXX plain vanilla data shown in Figure~\ref{fig: VXX model} is an empirical evidence that we can safely use such regularization.

\begin{figure}
    \begin{center}
    \scalebox{0.95}{%
    \begin{tikzpicture}
    \begin{axis}[no markers,
                xlabel=Moneyness,
                ylabel=Volatility,
                ylabel style={overlay},
                restrict x to domain=0.5:4.5,
                xmin=0.4, xmax=4.6,
                xtick={1,2,3,4},
                xticklabels={100\%,200\%,300\%,400\%},
                ymin=0, ymax=3.8,
                grid=major,
                axis background/.style={fill=gray!10}]
    \addplot [color=blue,thick,smooth] table [y=vixLVa,x expr=\thisrow{vixKa}/14.60] from \vixvxx;
    \addplot [name path=upper,color=mydarkred,thin] table [y expr=\thisrow{vixMKTa}+\thisrow{vixBAa}/2,x expr=\thisrow{vixKa}/14.60] from \vixvxx;
    \addplot [name path=lower,color=mydarkred,thin] table [y expr=\thisrow{vixMKTa}-\thisrow{vixBAa}/2,x expr=\thisrow{vixKa}/14.60] from \vixvxx;
    \addplot [fill=mydarkred!25] fill between[of=upper and lower];  
    \end{axis}
    \end{tikzpicture}}
    \hspace*{0.5cm}
    \scalebox{0.95}{%
    \begin{tikzpicture}
    \begin{axis}[no markers,
                xlabel=Moneyness,
                ylabel=Volatility,
                ylabel style={overlay},
                restrict x to domain=0.5:4.5,
                xmin=0.4, xmax=4.6,
                xtick={1,2,3,4},
                xticklabels={100\%,200\%,300\%,400\%},
                ymin=0, ymax=3.8,
                grid=major,
                axis background/.style={fill=gray!10}]
    \addplot [color=blue,thick,smooth] table [y=vixLVb,x expr=\thisrow{vixKb}/16.15] from \vixvxx;
    \addplot [name path=upper,color=mydarkred,thin] table [y expr=\thisrow{vixMKTb}+\thisrow{vixBAb}/2,x expr=\thisrow{vixKb}/16.15] from \vixvxx;
    \addplot [name path=lower,color=mydarkred,thin] table [y expr=\thisrow{vixMKTb}-\thisrow{vixBAb}/2,x expr=\thisrow{vixKb}/16.15] from \vixvxx;
    \addplot [fill=mydarkred!25] fill between[of=upper and lower];  
    \end{axis}
    \end{tikzpicture}}\\
    \scalebox{0.95}{%
    \begin{tikzpicture}
    \begin{axis}[no markers,
                xlabel=Moneyness,
                ylabel=Volatility,
                ylabel style={overlay},
                restrict x to domain=0.5:4.5,
                xmin=0.4, xmax=4.6,
                xtick={1,2,3,4},
                xticklabels={100\%,200\%,300\%,400\%},
                ymin=0, ymax=3.8,
                grid=major,
                axis background/.style={fill=gray!10}]
    \addplot [color=blue,thick,smooth] table [y=vixLVc,x expr=\thisrow{vixKc}/17.45] from \vixvxx;
    \addplot [name path=upper,color=mydarkred,thin] table [y expr=\thisrow{vixMKTc}+\thisrow{vixBAc}/2,x expr=\thisrow{vixKc}/17.45] from \vixvxx;
    \addplot [name path=lower,color=mydarkred,thin] table [y expr=\thisrow{vixMKTc}-\thisrow{vixBAc}/2,x expr=\thisrow{vixKc}/17.45] from \vixvxx;
    \addplot [fill=mydarkred!25] fill between[of=upper and lower];  
    \end{axis}
    \end{tikzpicture}}
    \hspace*{0.5cm}
    \scalebox{0.95}{%
    \begin{tikzpicture}
    \begin{axis}[no markers,
                xlabel=Moneyness,
                ylabel=Volatility,
                ylabel style={overlay},
                restrict x to domain=0.5:4.5,
                xmin=0.4, xmax=4.6,
                xtick={1,2,3,4},
                xticklabels={100\%,200\%,300\%,400\%},
                ymin=0, ymax=3.8,
                grid=major,
                axis background/.style={fill=gray!10}]
    \addplot [color=blue,thick,smooth] table [y=vixLVd,x expr=\thisrow{vixKd}/18.15] from \vixvxx;
    \addplot [name path=upper,color=mydarkred,thin] table [y expr=\thisrow{vixMKTd}+\thisrow{vixBAd}/2,x expr=\thisrow{vixKd}/18.15] from \vixvxx;
    \addplot [name path=lower,color=mydarkred,thin] table [y expr=\thisrow{vixMKTd}-\thisrow{vixBAd}/2,x expr=\thisrow{vixKd}/18.15] from \vixvxx;
    \addplot [fill=mydarkred!25] fill between[of=upper and lower];  
    \end{axis}
    \end{tikzpicture}}
    \end{center}
    \caption{VIX futures model implied volatility (blue solid line) against market bid-ask quotes (red shaded area), as on November 7, 2019, for the four maturities November 20, 2019, December 18, 2019, January 22, 2020 and February 19, 2020, from the top to the bottom. }
    \label{fig: VIX model}
\end{figure}
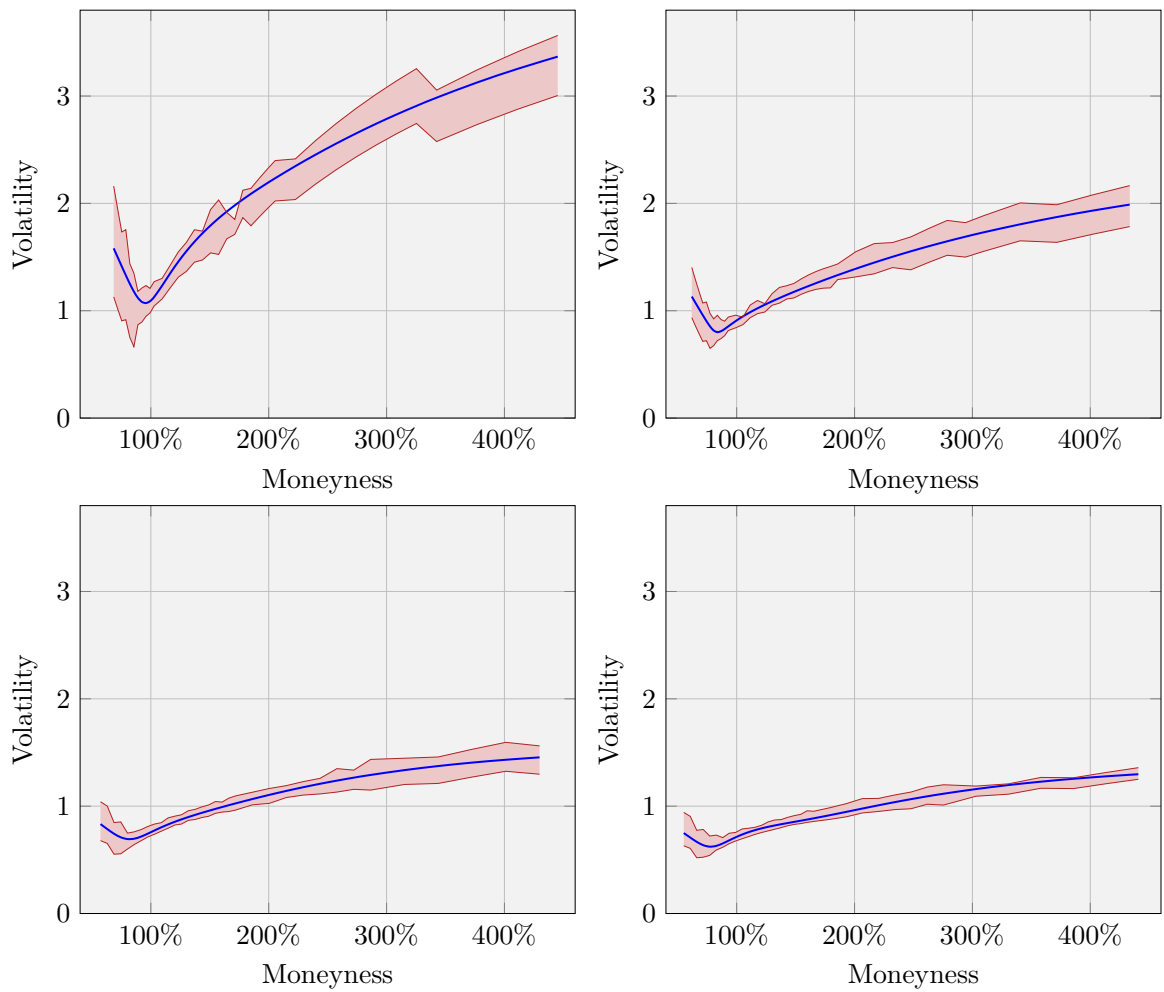

\begin{figure}
    \begin{center}
    \scalebox{0.95}{%
    \begin{tikzpicture}
    \begin{axis}[no markers,
                xlabel=Moneyness,
                ylabel=Volatility,
                ylabel style={overlay},
                restrict x to domain=0.3:2.7,
                xmin=0.2, xmax=2.8,
                xtick={0.5,1,1.5,2,2.5},
                xticklabels={50\%,100\%,150\%,200\%,250\%},
                ymin=0, ymax=3.5,
                grid=major,
                axis background/.style={fill=gray!10}]
    \addplot [color=blue,thick,smooth] table [y=vxxLVa,x expr=\thisrow{vxxKa}/19.22] from \vixvxx;
    \addplot [name path=upper,color=mydarkred,thin] table [y expr=\thisrow{vxxMKTa}+\thisrow{vxxBAa}/2,x expr=\thisrow{vxxKa}/19.22] from \vixvxx;
    \addplot [name path=lower,color=mydarkred,thin] table [y expr=\thisrow{vxxMKTa}-\thisrow{vxxBAa}/2,x expr=\thisrow{vxxKa}/19.22] from \vixvxx;
    \addplot [fill=mydarkred!25] fill between[of=upper and lower];  
    \end{axis}
    \end{tikzpicture}}
    \hspace*{0.5cm}
    \scalebox{0.95}{%
    \begin{tikzpicture}
    \begin{axis}[no markers,
                xlabel=Moneyness,
                ylabel=Volatility,
                ylabel style={overlay},
                restrict x to domain=0.3:2.7,
                xmin=0.2, xmax=2.8,
                xtick={0.5,1,1.5,2,2.5},
                xticklabels={50\%,100\%,150\%,200\%,250\%},
                ymin=0, ymax=3.5,
                grid=major,
                axis background/.style={fill=gray!10}]
    \addplot [color=blue,thick,smooth] table [y=vxxLVb,x expr=\thisrow{vxxKb}/19.22] from \vixvxx;
    \addplot [name path=upper,color=mydarkred,thin] table [y expr=\thisrow{vxxMKTb}+\thisrow{vxxBAb}/2,x expr=\thisrow{vxxKb}/19.22] from \vixvxx;
    \addplot [name path=lower,color=mydarkred,thin] table [y expr=\thisrow{vxxMKTb}-\thisrow{vxxBAb}/2,x expr=\thisrow{vxxKb}/19.22] from \vixvxx;
    \addplot [fill=mydarkred!25] fill between[of=upper and lower];  
    \end{axis}
    \end{tikzpicture}}\\
    \scalebox{0.95}{%
    \begin{tikzpicture}
    \begin{axis}[no markers,
                xlabel=Moneyness,
                ylabel=Volatility,
                ylabel style={overlay},
                restrict x to domain=0.3:2.7,
                xmin=0.2, xmax=2.8,
                xtick={0.5,1,1.5,2,2.5},
                xticklabels={50\%,100\%,150\%,200\%,250\%},
                ymin=0, ymax=3.5,
                grid=major,
                axis background/.style={fill=gray!10}]
    \addplot [color=blue,thick,smooth] table [y=vxxLVc,x expr=\thisrow{vxxKc}/19.22] from \vixvxx;
    \addplot [name path=upper,color=mydarkred,thin] table [y expr=\thisrow{vxxMKTc}+\thisrow{vxxBAc}/2,x expr=\thisrow{vxxKc}/19.22] from \vixvxx;
    \addplot [name path=lower,color=mydarkred,thin] table [y expr=\thisrow{vxxMKTc}-\thisrow{vxxBAc}/2,x expr=\thisrow{vxxKc}/19.22] from \vixvxx;
    \addplot [fill=mydarkred!25] fill between[of=upper and lower];  
    \end{axis}
    \end{tikzpicture}}
    \end{center}
\caption{VXX model implied volatility (blue solid line) against market bid-ask quotes (red shaded area), as on November 7, 2019, for the three maturities November 15, 2019, December 20, 2019 and January 17, 2020, from the top to the bottom. }
\label{fig: VXX model}
\end{figure}
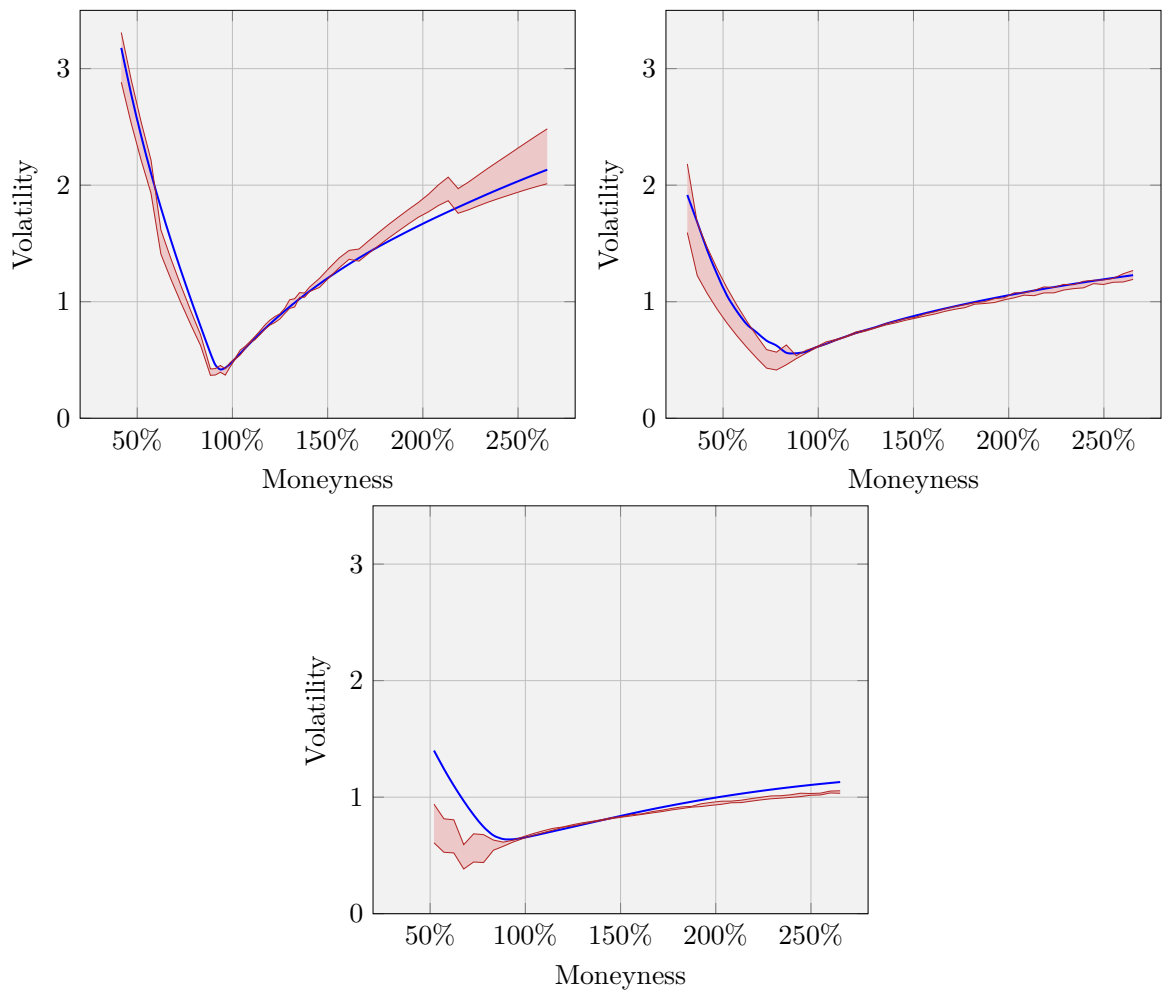

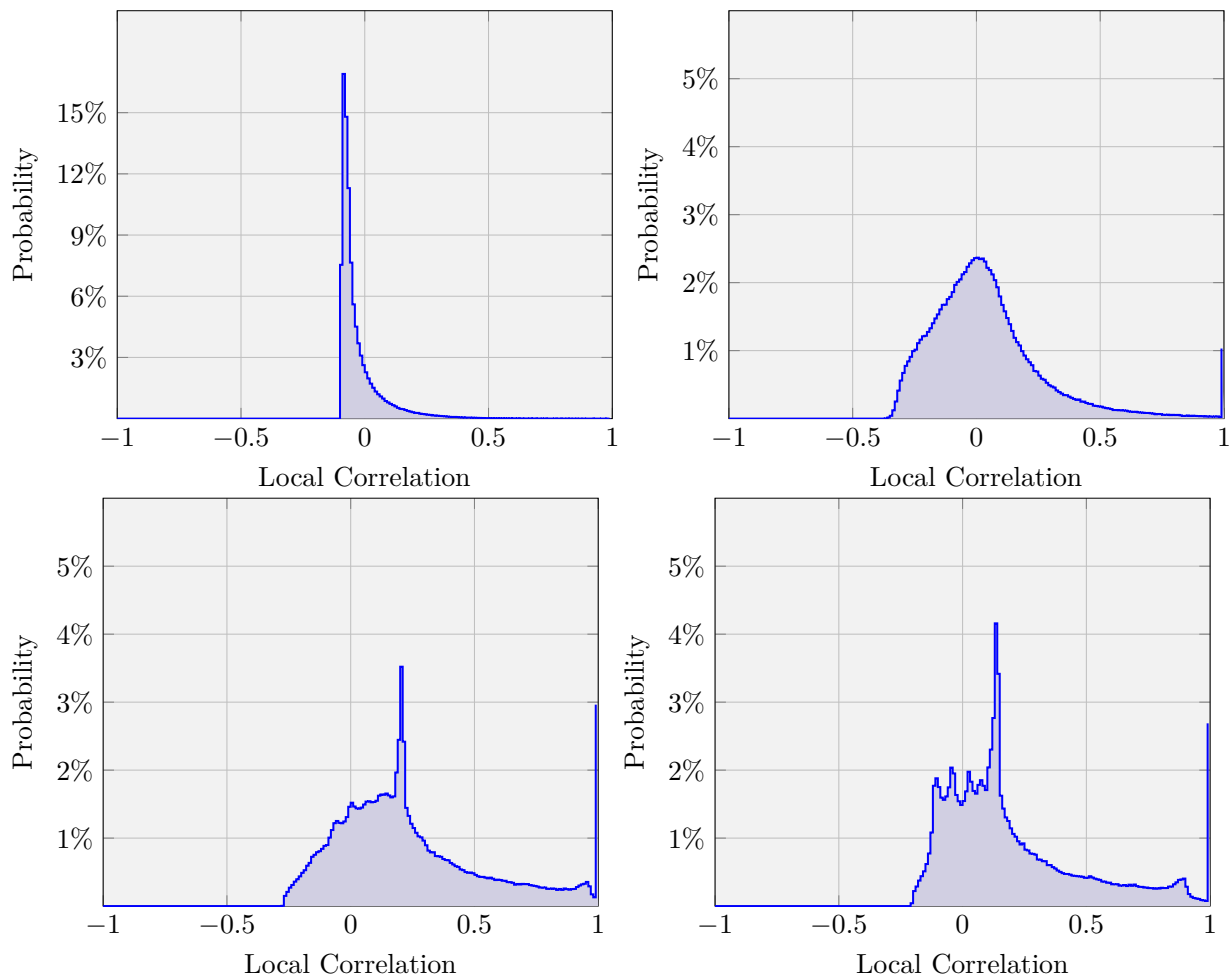
\begin{figure}
    \begin{center}
    \scalebox{0.95}{%
    \begin{tikzpicture}
    \begin{axis}[xlabel=Local Correlation,
                 ylabel=Probability,
                 ylabel style={overlay},
                 xmin=-1, xmax=1,
                 ymin=0, ymax=0.2,
                 scaled y ticks=false,
                 ytick={0.03,0.06,0.09,0.12,0.15},
                 yticklabels={3\%,6\%,9\%,12\%,15\%},
                 grid=major,
                 axis background/.style={fill=gray!10}]
    \addplot [name path=density,color=blue,thick,smooth,const plot mark left,mark=none] table [y=t02,x=xBin] from \rhohist;
    \path[name path=axis] (axis cs:-1,0) -- (axis cs:1,0);
    \addplot [fill=mydarkblue!25] fill between[of=density and axis];  
    \end{axis}
    \end{tikzpicture}}%
    \hspace*{0.5cm}
    \scalebox{0.95}{%
    \begin{tikzpicture}
    \begin{axis}[xlabel=Local Correlation,
                 ylabel=Probability,
                 ylabel style={overlay},
                 xmin=-1, xmax=1,
                 ymin=0, ymax=0.06,
                 scaled y ticks=false,
                 ytick={0.01,0.02,0.03,0.04,0.05},
                 yticklabels={1\%,2\%,3\%,4\%,5\%},
                 grid=major,
                 axis background/.style={fill=gray!10}]
    \addplot [name path=density,color=blue,thick,smooth,const plot mark left,mark=none] table [y=t09,x=xBin] from \rhohist;
    \path[name path=axis] (axis cs:-1,0) -- (axis cs:1,0);
    \addplot [fill=mydarkblue!25] fill between[of=density and axis];  
    \end{axis}
    \end{tikzpicture}}\\
    \scalebox{0.95}{%
    \begin{tikzpicture}
    \begin{axis}[xlabel=Local Correlation,
                 ylabel=Probability,
                 ylabel style={overlay},
                 xmin=-1, xmax=1,
                 ymin=0, ymax=0.06,
                 scaled y ticks=false,
                 ytick={0.01,0.02,0.03,0.04,0.05},
                 yticklabels={1\%,2\%,3\%,4\%,5\%},
                 grid=major,
                 axis background/.style={fill=gray!10}]
    \addplot [name path=density,color=blue,thick,smooth,const plot mark left,mark=none] table [y=t59,x=xBin] from \rhohist;
    \path[name path=axis] (axis cs:-1,0) -- (axis cs:1,0);
    \addplot [fill=mydarkblue!25] fill between[of=density and axis];  
    \end{axis}
    \end{tikzpicture}}%
    \hspace*{0.5cm}
    \scalebox{0.95}{%
    \begin{tikzpicture}
    \begin{axis}[xlabel=Local Correlation,
                 ylabel=Probability,
                 ylabel style={overlay},
                 xmin=-1, xmax=1,
                 ymin=0, ymax=0.06,
                 scaled y ticks=false,
                 ytick={0.01,0.02,0.03,0.04,0.05},
                 yticklabels={1\%,2\%,3\%,4\%,5\%},
                 grid=major,
                 axis background/.style={fill=gray!10}]
    \addplot [name path=density,color=blue,thick,smooth,const plot mark left,mark=none] table [y=t67,x=xBin] from \rhohist;
    \path[name path=axis] (axis cs:-1,0) -- (axis cs:1,0);
    \addplot [fill=mydarkblue!25] fill between[of=density and axis];  
    \end{axis}
    \end{tikzpicture}}%
    \end{center}
    \caption{Histograms of the distribution of the local correlation parameter $\rho(t,V_t)$ for $t=2,9,59,67$ days.}
\label{fig:histo}
\end{figure}

\section{Conclusion and further developments}
\label{conclusion}

In this paper, we have provided a general framework for a calibration of Exchange-Traded Products (ETP) based on futures strategies. In a numerical simulation based on real data we have considered, as an example of ETP, the VXX ETN, which is a strategy on the nearest and second-nearest maturing VIX futures contracts. The main ingredient to achieve our goal  was a full stochastic local-volatility model that can be calibrated on a book of options on a general ETP and its underlying futures contracts in a parsimonious manner.

A parameters sensitivity analysis allowed us to fix suitable values for the model parameters in order to fit the smile of the VIX and VXX ETN, with a high level of accuracy. Our specification led to an admissible local correlation function that justifies the consistency of the whole procedure. Remarkably, we found that the stochastic component of the volatility was necessary in order to fit market quotes. In other words, a purely local volatility approach is not rich enough for calibrating both the VIX and VXX markets. The implementation of a full calibration algorithm (in order to find the optimal value of all model parameters) based on neural-network techniques, in the same spirit of \cite{Avellaneda2021}, is currently under investigation. The methodology proposed has been applied to the case of VIX futures and VXX ETN, but the framework is flexible enough to tackle also more general ETP's. Last but not least, the model is based on a simple diffusive dynamics for the VIX and it can be easily extended in order to include jumps, regime switching and other features that could be required for describing the stylized facts of this market.

\bibliography{biblio}
\bibliographystyle{apalike}

\appendix
\section{Alternative specification for futures volatilities}
\label{appendix} 

Here, we describe an alternative specification for the futures volatility process given by Equation~\eqref{Process nu} by promoting the driving scalar process to be a matrix process.

We start with a new definition of the futures dynamics, given by Equation~\eqref{SLV model for F}, which allows us to introduce the matrix driving process.
\begin{equation*}
    dF_t^i = \trace{ \nu_t^i \cdot  dW_t },
\end{equation*}%
where $W_t$ is a standard Brownian motion with values in the space of $d\times d$ matrices, and the matrix volatility process $\nu_t^i$  is defined as
\begin{equation*}
    \nu_t^i \doteq \ell_F^i(t, F_t^i) \, \sqrt{v_t} \cdot R_i(t, V_t),
\end{equation*}%
where the driving process $v_t$ has values in the space of $d\times d$ positive-semidefinite symmetric matrices. while the local-correlation processes $R_i(t, V_t)$ has values in the space of $d\times d$ orthogonal matrices, so that $R_i(t,V_t) \cdot R_i^{\top}(t,V_t)$ is the $d\times d$ identity matrix.

We apply the  Gy{\"o}ngy Lemma in order to obtain the analogue of~\eqref{leverage function}, and thanks to the fact that for any $d\times d$ matrix $A,B,$ we have that $\trace{A \cdot dW_t} \trace{B \cdot dW_t} = \trace{A \cdot B^{\top}} \,dt$, we get
\begin{equation*}
    \ell_F^i(t,K) = \frac{ \eta_F^i(t,K) }{ \sqrt{\ExC{0}{\trace{v_t}}{ F_t^i = K }} }.
\end{equation*}
On the other hand, for the ETP strategy, we obtain
\begin{align*}
    \eta_V^2(t,K)
    &= \ExC{0}{ \sum_{ij=1}^N \hat{\omega_t}^i \hat{\omega}_t^j \trace{ \nu_t^i \cdot (\nu_t^j)^{\top} } }{ V_t=K },  \\
    &= \ExC{0}{ \sum_{ij=1}^N \hat{\omega}_t^i \hat{\omega}_t^j \ell_F^i(t,F_t^i) \ell_F^j(t,F_t^j) \trace{ \sqrt{v_t} \cdot R_i(t, V_t) \cdot R_j^{\top}(t, V_t) \cdot \sqrt{v_t} } }{ V_t=K }\\
    &= \sum_{ij=1}^N \trace{ R_i(t, K) \cdot R_j^{\top}(t, K) \cdot \ExC{0}{ \hat{\omega}_t^i \hat{\omega}_t^j \ell_F^i(t,F_t^i) \ell_F^j(t,F_t^j) \, v_t }{ V_t=K } }.
\end{align*}

We now focus on the VIX/VXX case and we select a suitable definition of the local-correlation process to recover the analogue of the result in Proposition~\ref{Prop: rho computed VXX}. We assume $d=N$, where $N$ is the number of the futures contracts, and we search for a local correlation process satisfying the following constraints: $R_i(t,K) \cdot R_i^{\top}(t,K) = 1$ and $\trace{ R_i(t,K) \cdot R_j^{\top}(t,K) } = \rho(t,K)$, where $\rho(t,K)$ is a scalar function of time and price. Here, we build a solution for the case $d=N=2$. We define the local correlations matrices $R_i(t,K)$ with $i=1,2$ as given by
\begin{equation*}
    R_1(t,K) \doteq \begin{bmatrix} 1 & 0 \\ 0 & 1 \end{bmatrix}
    \;,\quad
    R_2(t,K) \doteq \begin{bmatrix} \rho(t, K) & - \sqrt{1 - \rho^2(t, K)} \\ \sqrt{1 - \rho^2(t, K)} & \rho(t, K) \end{bmatrix}.
\end{equation*}%
In this case we can recover the analogue of the result in Proposition~\ref{Prop: rho computed VXX}. In fact, from 
\begin{align*}
    \eta_V^2(t,K)
    & =  \ExC{0}{ \left( \hat{\omega}_t^{(1)} \ell_F^{(1)}(t,F_t^{(1)}) \right)^2 \trace{ v_t } + \left( \hat{\omega}_t^{2} \ell_F^{(2)}(t,F_t^{(2)}) \right)^2 \trace{ v_t } }{ V_t = K } \\
    & \quad + 2 \ExC{0}{ \hat{\omega}_t^{(1)} \hat{\omega}_t^{(2)} \ell_F^{(1)}(t,F_t^{(1)}) \ell_F^{(2)}(t,F_t^{(2)}) \, \rho(t,V_t) \,\trace{ v_t } }{ V_t = K },
\end{align*}%
we can define 
\begin{equation*}
     A^{(1)}(t,K) \doteq \ExC{0}{ \left( \hat{\omega}_t^{(1)} \ell_F^{(1)}(t,F_t^{(1)}) \right)^2 \trace{ v_t } }{ V_t = K },
 \end{equation*}
\begin{equation*}
     A^{(2)}(t,K) \doteq \ExC{0}{ \left( \hat{\omega}_t^{(2)} \ell_F^{(2)}(t,F_t^{(2)}) \right)^2 \trace{ v_t } }{ V_t = K },
\end{equation*}
\begin{equation*}
      A^{(12)}(t,K) \doteq \ExC{0}{  \hat{\omega}_t^{(1)} \hat{\omega}_t^{(2)} \ell_F^{(1)}(t,F_t^{(1)}) \ell_F^{(2)}(t,F_t^{(2)}) \,\trace{ v_t } }{ V_t = K },
 \end{equation*}
and we get the analogue of \eqref{Rho equation}, namely
\begin{equation*}
      \rho(t,K) = \frac{\eta_V^2(t, K) - A^{(1)}(t, K ) - A^{(2)}(t, K)}{2 A^{(12)}(t, K)}.    
\end{equation*}
Of course, a matrix specification for $R_i(t,  V_t)$ leads to a richer (though less parsimonious) correlation structure.

\end{document}